\documentclass[11pt,twoside, letter]{article}
\usepackage{amsthm, amsmath, amssymb, amsfonts, graphicx, epsfig}
\usepackage{accents}
\usepackage{bbm}
\usepackage{mathrsfs}
\usepackage{cancel}
\usepackage{epstopdf}
\usepackage{graphicx}
\usepackage[font=sl,labelfont=bf]{caption}
\usepackage{subcaption}
\usepackage{float}
\restylefloat{table}
\usepackage{slashbox}  %Removed by Ares

\usepackage{enumerate}
\usepackage[usenames,dvipsnames]{color}

\usepackage[colorlinks=true, pdfstartview=FitV, linkcolor=blue,
            citecolor=blue, urlcolor=blue]{hyperref}
\usepackage[usenames]{color}
\usepackage{url}
\makeatletter\def\url@leostyle{%
 \@ifundefined{selectfont}{\def\UrlFont{\sf}}{\def\UrlFont{\scriptsize\ttfamily}}} \makeatother
\usepackage[margin=1.3in, letterpaper]{geometry}

\newtheorem{theorem}{Theorem}

\newtheorem{lemma}[theorem]{Lemma}

\theoremstyle{definition}

\newtheorem{example}[theorem]{Example}
\theoremstyle{remark}
\newtheorem{remark}[theorem]{Remark}

\numberwithin{equation}{section}
\numberwithin{theorem}{section}
\definecolor{Red}{rgb}{0.9,0,0.0}
\definecolor{Blue}{rgb}{0,0.0,1.0}
\def\cA{\mathcal{A}}

\def\cD{\mathcal{D}}

\def\cG{\mathcal{G}}

\def\cI{\mathcal{I}}

\def\cQ{\mathcal{Q}}
\def\cR{\mathcal{R}}

\def\cT{\mathcal{T}}

\def\bR{\mathbb{R}}

\def\sF{\mathscr{F}}
\def\sG{\mathscr{G}}

\def\mP{\mathsf{P}}

\newcommand{\1}{\mathbbm{1}}            % preferable way of writing indicator function
\newcommand{\set}[1]{\{#1\}}            % set: {xyz} to be used for inline formulas
\newcommand{\Set}[1]{\left\{#1\right\}} % set: {xyz} to be used for seapare (not inline) formulas
\renewcommand{\mid}{\;|\;}              % mid bar with small spaces before and after: x | y
\newcommand{\Mid}{\;\Big | \;}          % big bar with small spaces before and after:
\DeclareMathOperator*{\esssup}{ess\,sup} % ess sup
\DeclareMathOperator*{\essinf}{ess\,inf} % ess inf

\DeclareMathOperator{\var}{\mathrm{V}@\mathrm{R}}           % \V@R Value-at-risk
\DeclareMathOperator{\avar}{\mathrm{AV}@\mathrm{R}}         % \AV@R average Value-at-risk
\usepackage[nonumberlist]{glossaries}

\glsdisablehyper
\newglossaryentry{T}{name={\ensuremath{T}}, description={relevant time horizon for CCP}}
\newglossaryentry{I}{name={\ensuremath{I}}, description={number of CMs in CCP}}
\newglossaryentry{deltaf}{name={\ensuremath{\delta_f}}, description={fundamental unit of time. Also, time between IM and VM calls}}
\newglossaryentry{delta}{name={\ensuremath{\delta}}, description={margin period of risk}}
\newglossaryentry{beta}{name={\ensuremath{\beta}}, description={discount factor}}
\newglossaryentry{deltac}{name={\ensuremath{\delta_c}}, description={$:=n\delta_f+\delta$, default time of CCP}}
\newglossaryentry{tk}{name={\ensuremath{ t_k}}, description={$:=k\delta_f, \, k=0,1,\ldots$. }}
\newglossaryentry{DeltaF}{name={\ensuremath{\Delta_f}}, description={time between DF calls}}
\newglossaryentry{Tj}{name={ \ensuremath{T_j}}, description={$:=j\Delta_f, \, j=0,1,\ldots$}}
\newglossaryentry{cTf}{name={\ensuremath{\mathcal{T}^f}}, description={$:=\set{t_0,t_1,\ldots,T-\delta_f}$}. IM and VM call times}
\newglossaryentry{cT}{name={\ensuremath{\mathcal{T}}}, description={$:=\set{0,\delta_f, 2\delta_f,\ldots, T-\delta_f,T}$}}
\newglossaryentry{cTF}{name={\ensuremath{\mathcal{T}^F}}, description={$:=\set{T_0,T_1,\ldots,T-\delta_f}$}. DF call times}
\newglossaryentry{cTdelta}{name={\ensuremath{\mathcal{T}^\delta}},
                                description={$:=\set{t_1+\delta,t_2+\delta,\ldots}$}. Unfunded DF calls (if needed) times}

\newglossaryentry{Vti}{name={\ensuremath{V_t^i}}, description={nominal (MTM) portfolio value at time $t$ of the CM$^i$}}
\newglossaryentry{Di}{name={\ensuremath{D^i}}, description={dividend payment associated with the portfolio of the CM$^i$}}
\newglossaryentry{DFi}{name={\ensuremath{\text{DF}^i}}, description={the DF contribution of the CM$^i$}}
\newglossaryentry{VMitk}{name={\ensuremath{\textrm{VM}^i_{t_k}}}, description={variation margin of the $i$ member at time $t_k$}}
\newglossaryentry{IMitk}{name={\ensuremath{\textrm{IM}^i_{t_k}}}, description={initial margin of the CM$^i$ at time $t_k$}}
\newglossaryentry{Xitk}{name={\ensuremath{X^i_{t_k}}}, description={cash flow based on which the $\textrm{IM}^i_{t_k}$ is computed}}

\newglossaryentry{rho}{name={\ensuremath{\rho}}, description={dynamic risk measure used for compuations of IM}}
\newglossaryentry{eta}{name={\ensuremath{\eta}}, description={dynamic risk measure used for compuations of total pre-funded DF}}
\newglossaryentry{taui}{name={\ensuremath{\tau_i}}, description={default time of the CM$^i$}}
\newglossaryentry{hatVi}{name={\ensuremath{\widehat V^i}},
                                description={the recovered value  of the the portfolio of the CM$^i$ by means of liquidation or auctioning}}
\newglossaryentry{EPi}{name={\ensuremath{\textrm{EP}^i}}, description={net exposure of the CCP to the default of CM$^i$}}
\newglossaryentry{EPiReg}{name={\ensuremath{\textrm{EP}^{i,\textrm{reg}}}}, description={regulatory net exposure of the CCP to the default of CM$^i$}}
\newglossaryentry{CMi}{name={\ensuremath{\textrm{CM}^i}}, description={the $i$th clearing member}}
\newglossaryentry{Tdelta}{name={\ensuremath{T^\delta}}, description={latest possible time for calling uDF}}
\newglossaryentry{SG}{name={\ensuremath{\textrm{SG}}}, description={CCP equity or skin-in-the-game}}
\newglossaryentry{uDFi}{name={\ensuremath{\text{uDF}^i}}, description={Unfunded Default Fund of \gls{CMi}}}
\newglossaryentry{EL}{name={\ensuremath{\text{EL}}}, description={Effective CCP loss}}

\newglossaryentry{Ri}{name={\ensuremath{R^i}}, description={Credit ratings process of \gls{CMi}}}
\newglossaryentry{cRi}{name={\ensuremath{\cR^i}}, description={$:=\set{1,\ldots,K^i}$. The state space of the credit ratings process \gls{Ri}}}
\newglossaryentry{mPi}{name={\ensuremath{\mP^i}}, description={Transition matrix of process $R^i$}}
\newglossaryentry{mP}{name={\ensuremath{\mP}}, description={Transition matrix of process $R$}}
\newglossaryentry{R}{name={\ensuremath{R}}, description={process modeling the joint evolution of credit ratings of all CMs}}
\newglossaryentry{mPiy}{name={\ensuremath{\mP^{i,y}}}, description={One year transition matrix for credit ratings of \gls{CMi}}}

\setacronymstyle{long-short}
\newacronym{OTC}{OTC}{Over The Counter}
\newacronym[longplural={Central Clearing Parties}]{CCP}{CCP}{Central Clearing Party}
\newacronym{IM}{IM}{Initial Margin}
\newacronym{VM}{VM}{Variation Margin}
\newacronym{DF}{DF}{Default Fund}
\newacronym{CM}{CM}{Clearing Member}
\newacronym{MtM}{MtM}{Mark-to-Market}
\newacronym{VaR}{\ensuremath{\var}}{Value at Risk}
\newacronym{AVaR}{\ensuremath{\avar}}{Average Value at Risk}
\newacronym{ES}{ES}{Expected Shortfall}
\newacronym{EMIR}{EMIR}{European Market Infrastructure Regulation}
\newacronym{CDS}{CDS}{Credit Default Swap}
\newacronym{IRS}{IRS}{Interest Rate Swap}
\newacronym{BCBS}{BCBS}{Basel Committee on Banking Supervision}
\newacronym{uDF}{uDF}{Unfunded Default Fund}

\makeglossaries

\title{A Dynamic Model of Central Counterparty Risk}

\makeatletter
\def\and{%
  \end{tabular}%
  \begin{tabular}[t]{c}}%
\def\@fnsymbol#1{\ensuremath{\ifcase#1\or a\or b\or c\or
   d\or e\or f\or g\or h\or i\else\@ctrerr\fi}}
\makeatother

\author{
Tomasz R. Bielecki\,\thanks{Department of Applied Mathematics, Illinois Institute of Technology
       \newline \hspace*{1.45em}  10 W 32nd Str, Building REC, Room 208, Chicago, IL 60616, USA
       \newline \hspace*{1.45em}  Emails: \url{tbielecki@iit.edu} (TRB), \url{cialenco@iit.edu} (IC), and \url{sfeng10@hawk.iit.edu} (SF)
       \newline \hspace*{1.45em}  URLs: \url{http://math.iit.edu/\~bielecki} (TRB), and \url{http://math.iit.edu/\~igor} (IC)
        \vspace{0.5em}} ,
\and
         Igor Cialenco,\,\footnotemark[1] \newline
\and
        Shibi Feng,\,\footnotemark[1] \newline
}        %}

\date{ \small  %\today\\ Preliminary Draft}} %
First Circulated: February 18, 2018
}

\begin{document}

\maketitle

{\footnotesize
\begin{tabular}{l@{} p{350pt}}
  \hline \\[-.2em]
  \textsc{Abstract}: \ &
We introduce a dynamic model of the default waterfall of derivatives CCPs and propose a risk sensitive method for sizing the initial margin (IM), and the default fund (DF) and its allocation among clearing members. Using a Markovian structure model of joint credit migrations, our evaluation of DF takes into account the joint credit quality of clearing members as they evolve over time. Another important aspect of the proposed methodology is the use of the time consistent dynamic risk measures for computation of IM and DF.
We carry out a comprehensive numerical study, where, in particular, we analyze the advantages of the proposed methodology and its comparison with the currently prevailing methods used in industry.
\\[0.5em]
\textsc{Keywords:} \ &  Central Clearing Parties; CCP; default fund; initial margin; dynamic risk measure; Markov structures; CDS; default waterfall; credit migration\\
%\textsc{MSC2010:} \ &  \\[1em]
\textsc{JEL Codes:} \ & G01, G18, G20, G23, G28 \\[1em]
  \hline
\end{tabular}
}

\bigskip

%\tableofcontents

\section{Introduction}

Following the aftermath of the financial crisis of 2008, starting with 2009 G-20 clearing mandate,  regulators across the globe required that standardized \gls{OTC} derivative contracts are cleared by the \glspl{CCP}. In US the reforms have been implemented through  Dodd--Frank Wall Street Reform and Consumer Protection Act~\cite{DoodFrank} (see also \cite{FSB2014}), and in Europe through \gls{EMIR}~\cite{EMIR2012}. As such, the entire landscape of financial markets has been changing, and CCPs are now recognized, by practitioners and regulators, to be a vital and critical element of the financial systems. It is well recognized that understanding how CCPs and their structure will influence the financial markets and stability of the financial system represents one of the key  challenges faced today by the market participants.
In a nutshell, the CCPs are intended to mitigate the counterparty risk, increase transparency, facilitate regulatory access to the necessary data, protect against market abuse, and avoid contagion if one of the (large) financial institution defaults; for more details on how CCPs can increase the safety and integrity of financial markets see for instance  \cite{EUREX2014}.

A CCP brings market together, and acts as a buyer for every seller, and a seller for every buyer. Every OTC contract\footnote{For example, the \gls{CDS})  are cleared by CME and ICE Clear,  \gls{IRS} are cleared by LCH.Clearnet and CME.} between two counterparties is replaced (through the novation process) by two contracts between CCP and each counterparty, such that one contract offsets the other. If one counterparty defaults, then the other is protected by the \textit{default management procedure} and resources of CCP, as described bellow. This means that CCP is exposed to the risk that one of the counterparties will default. Hence, the CCP becomes an important systemic element of the modern financial industry, and it is critically important that the CCP manages appropriately the risk.

The literature regarding various aspects of CCPs activities, in particular the risk management methodologies, has been rapidly growing in the recent years. We refer the reader to some recent studies, e.g., \cite{ArmentiCrepey2015,Arnold2017,BeliVaaradi2017,CapponiEtAl2017,ChengPhd2017,CuiEtAl2017,Deng2017,YoungPaddrik2017,GhamamiGlasserman2016}, and the references therein.

The main goal of this paper is to develop a new methodology for computing some of the ingredients of the CCP's default waterfall, in a dynamic framework, and to test it via a numerical study.
In particular, we propose a novel, risk sensitive  method for computing the total default fund of CCP.
Risk sensitivity amounts to accounting for credit migrations of the clearing members and the stochastic dependence between these migrations, which is modeled in terms of Markov structures. Another important aspect of the proposed methodology is the use of the time consistent dynamic risk measures.

The major objective of the numerical study is to compute the \gls{DF}/\gls{IM} ratio for various model configurations.
For this purpose, we consider a stylised example of a CCP consisting of 8 clearing members, each holding a portfolio of Credit Default Swap (CDS) contracts.
The obtained numerical results indicate that our approach may offer a significant improvement over the existing methods of computing the elements of the default waterfall currently employed by CCPs.  In particular, our findings show that the DF/IM ratio does not vary with the number of members or the size of the positions holding by the members. This indicates that our model scales appropriately the sizes of both the DF and IM with the increasing number of CMs. On the contrary, the Cover~1/Cover~2 methodology (cf. \cite{CFTC2016}) for computing \gls{DF} does not scale appropriately the size of DF with the increasing number of CMs.

The paper is organized as follows. In Section~\ref{sec:mechanics} we briefly recall the mechanics of a CCP risk management and the structure of the risk waterfall. Section~\ref{sec:D-CCP} is dedicated to a dynamic model of a CCP operation. We devote a separate subsection to each layer of the default waterfall, which, in particular, includes the analysis of the proposed methods and their comparison with the existing ones.
There is no ambiguity on defining the variation margin, which is just mark-to-market of the member's portfolio. In Section~\ref{sec:IM} we define the cashflow on which the \gls{IM} is computed, and we use a dynamic convex risk measure to determine the \gls{IM}. The (prefunded) \gls{DF} is studied in Section~\ref{sec:DF}. We start by identifying the net exposure of the \glspl{CM}, and the computation of the total \gls{DF}, where we again make use of a dynamic convex risk measure. In Section~\ref{sec:DFAllocation} we investigate the allocation of the total \gls{DF} among the \glspl{CM}, and we propose a method rooted in the theory of capital allocation based on risk contributions, that uses the extreme measures from the robust robust representation of risk measures. Finally, in Section~\ref{sec:numerical} we care out a comprehensive numerical study, where, in particular, we describe the modeling aspects of the cleared products, the dependence structure of the defaults times of the \glspl{CM}, as well as the advantages of the proposed methodology and its comparison with the currently prevailing methods used in industry.
We deferred to appendix some technical aspects of the paper: Appendix~\ref{sec:avar} contains relevant material on Conditional \gls{AVaR}, Appendix~\ref{sec:MC} is devoted to Markovian structures, and in Appendix~\ref{sec:bigbang} we briefly present the models we use for modeling portfolios of \glspl{CDS} contracts.

\section{CCP Default Waterfall}\label{sec:mechanics}

A CCP is an independent entity, that complies with regulatory requirements, and clears certain standardized  financial contracts. The general guidelines and principles for regulatory capital are  given by \gls{BCBS}, and in US the CCPs are regulated and supervised by governmental entities such as FED, SEC, CFTC.
A financial institution that wants to clear/trade through a CCP must be a \gls{CM} of CCP. The number of members of a CCP varies: ICE Clear US has 28 {members}, OCC has 100+ etc.

We will briefly outline the structure of a typical CCP;  for more details on mechanics of CCPs see for instance the monographs \cite{Gregory2014Book,Murphy2013}. The members of the CCP clear through the CCP a portfolio of assets (usually of the same class).
To remain financially solvent the CCPs charges each member various `margins', and the default management procedure of a CCP is done through so called \textit{default waterfall} or \textit{loss waterfall} to cover the losses due to the default of the CMs.

Each member maintains a margin account with the CCP, that is replenished, typically on daily basis,  through \gls{VM}, which is the \gls{MtM} value  of the open positions. Thus, the daily changes of the \gls{MtM} of the member's portfolio is transferred to the CCP. In practice, only the changes beyond a given threshold are transferred, but in this study we will omit this technicality.

Additionally, the CCP charges each member an \gls{IM}, that aims to cover the risk exposure of the CCP arising from
potential future market fluctuations of the member's portfolio over some risk horizon or margin period of risk (typically 5-10 days). The IM is also usually called on daily basis.

On top of that, each CM contributes to the prefunded Default Fund that acts as a form of mutualised loss sharing. It is also called clearing deposits or guaranty fund contributions. For brevity we will call it simply the \gls{DF}. The \gls{DF} is usually called on monthly basis.

As an incentive to implement proper risk management, the CCP pledges part of its equity (about 20-25\%) to be used to cover losses above the IM and DF of the defaulted CMs, and sometimes additional capital if the total DF (of all CMs) was already used. Traditionally this is called CCP equity contribution or skin-in-the-game.

If all the above margins and layers of defence are insufficient, the CCP may call for additional capital contribution from the survived CMs, called \gls{uDF}.

Below, we give a schematic of the default waterfall. Each consecutive layer of the waterfall is used, if losses exceed the sum of funds from the previous layers:

\smallskip

 \begin{center}
IM of the defaulted member(s) \\
$\Downarrow$ \\
(pre-funded) DF contribution of the defaulted member(s) \\
$\Downarrow$ \\
CCP equity, or skin-in-the-game \\
 $\Downarrow$  \\
(pre-funded) DF of the surviving members \\
$\Downarrow$ \\
Unfunded DF from the surviving  members\\
and/or  some additional capital from CCP  \\
\end{center}
If the losses of CCP are still beyond the above waterfall layers, the CCP enters into recovery mode, by taking some extraordinary measures to cover the losses. Finally, if none of the above measures can bypass the overall losses, the CCP becomes insolvent and defaults, by getting through resolution plan.
As already mentioned, the capital structure of the CCPs, and the design of the default waterfall have to meet the regulatory requirements.

\section{A Dynamic Model of the Default Waterfall}\label{sec:D-CCP}
In this section we provide an extension of the static model of CCP that was put forth in \cite{Ghamami2014} to dynamic, discrete time setting.

\subsection{Generalities}\label{sec:generalities}
We consider a CCP consisting of $\gls{I}$ \glspl{CM}, with $I>1$. Our study regards a generic CCP without  making any specific assumptions or postulates on the nature of the cleared portfolios. For simplicity, it is assumed that the constituent portfolios of CMs do not change, unless defaults happen, and that no new clearing members are added. Hence, the number of \glspl{CM} may only decrease due to their defaults. Nevertheless, our model can easily be adopted to the general case, when new CMs may be added to the composition of the CCP.

As said, we work in the discrete time framework. Accordingly, we denote by $\gls{deltaf}>0$ the fundamental unit of time; typically, $\delta_f$ is one business day, and sometimes even smaller. Thus, in what follows all considered  time instances will be implicitly assumed to be multiples of $\delta_f$. We will make use of the notations $\gls{cT}:=\set{0,\delta_f,2\delta_f,\ldots, T-\delta_f,T}$, and $\gls{tk}:=k\delta_f, k=0,1,\ldots,$ where $T$ denotes the time horizon relevant to CCP.

Our model is risk-sensitive. By this we mean that the model accounts for possible credit migrations of the CMs. We postulate, that changes, or migrations, in credit ratings of the CMs, in particular their defaults,   can only occur at times in $\cT$.
We will denote by $\gls{taui}$ the default time of the CM$^i$.
In addition, our model allows for the default of the CCP. This aspect of the model will be analyzed in detail in  a future work.

Conforming to the general practice of CCP operations, we suppose that the IMs and the VMs are computed, and called,  according to discrete tenor  $\gls{cTf}:=\set{0,\delta_f,2\delta_f,\ldots, T-\delta_f}$, and that  the DF calls are made according to their respective tenor $t\in \gls{cTF}:=\set{T_0,T_1,T_2,\ldots , T-\Delta_f}$, where $\gls{Tj}:=j\Delta_f, j=0,1,\ldots$ with \gls{DeltaF} denoting a time interval, which is multiple of $\delta_f$, and usually of order of weeks.
The (potential) uDF calls are done, if needed,  according to a shifted tenor $\gls{cTdelta}:=\set{t_1+\delta, t_2+\delta,\ldots , T^\delta}$ for some $\gls{Tdelta}\leq T.$

We fix a filtered probability space $(\Omega,\sF,\mathbb{F},P)$, where the filtration $\mathbb{F}=(\sF_t, t\in \mathcal{T})$ models the information flow available to the CCP. We will denote by $E$ the expectation under probability $P$. The probability measure $P$ is interpreted as the statistical  (or the actuarial) probability measure.

We make a standing assumption that, for ech $i$, the default time  $\gls{taui}$  is a stopping time with respect to $\mathbb{F}$.

All random processes considered here are defined on  $(\Omega,\sF,\mathbb{F},P)$. In particular, the processes are adapted to $\mathbb{F}$.

We denote by $\gls{Vti}$ the nominal portfolio value at time $t\in\cT$ of the member $i\in \cI$, where by nominal we mean the MtM portfolio values,  given that the members do not default prior or at  $t\in \cT^{f}$. In addition, we denote by \gls{beta} the discount factor used by the CCP.

Throughout, we view all the cash flows from the perspective of the CCP, and thus we adopt the convention that a positive cash flow indicates an inflow of funds to the CCP from the member(s), and a negative cash flow is a payment by CCP to the member(s). In particular, since CCP runs a matched order book we must have that
$$
\sum_{i\in\cI}V^i_t=0.
$$

The dividend process associated with the $i$th CM portfolio is denoted by $\gls{Di}$, and thus the dividend payment at time $t_k$ by the member $i$ to the CCP is equal to $D^i_{t_k}$.

The dynamics of credit migrations of all the CMs are modeled in terms of a multivariate Markov chain, and are subject to marginal constraints; we refer to Appendix~\ref{sec:MC} for details.

We now turn to modeling the relevant margins.

%\subsection{Variation Margin}

\subsection{Variation Margin and Initial Margin}\label{sec:IM}
There is no ambiguity about the definition of the VM.  It is just the MtM of the member's portfolio, so that the VM of the $i$th member at time $t_k$  is given by
\begin{equation}\label{eq:VM}
\gls{VMitk}=V^i_{t_{k-1}}.
\end{equation}

The rest of this section will be devoted to the IM. Currently, the IM is usually computed as \gls{VaR} or as \gls{ES}, at some confidence level, applied to the exposure.

Given the dynamic setup adopted in this paper, we propose to use a general dynamic risk measure \cite{BCP2017,AcciaioPenner2010} to compute the IM through time.

Let us denote  by \gls{Xitk} the cash flow (adjusted for time value of money at time $t_k$), based on which the \gls{IMitk} is computed:
\begin{equation}\label{eq:Xitk}
X^i_{t_k} = \beta^{-1}_{t_k}\beta_{t_{k}+\delta}V^i_{ t_{k}+\delta} +\beta^{-1}_{t_k}\sum_{u=t_{k}}^{ t_{k}+\delta}\beta_{u}D^i_u - \textrm{VM}^i_{t_k}.
\end{equation}

According to our convention, a positive value of $ X^i_{t_k}$ is the net exposure of the CPP towards the member $i$.
Accordingly, we propose to compute the \gls{IM} called from the $i$th member at time $t_k$  as
\begin{equation}\label{eq:IMmethod2}
\textrm{IM}^i_{t_k} = \gls{rho}_{t_k}(-(X^i_{t_k})^+),
\end{equation}
where $\rho$ is a dynamic convex risk measure (a function that is local, monotone decreasing, cash-additive, convex, and time consistent; see Remark~\ref{rem:DRM}.(i) for more details).

Computation of $\textrm{IM}^i_{t_k} $ requires use of two different probability measures. When computing the value $V^i_{t_{k}+\delta}$ of the mark-to-market we need to use a risk neutral (or pricing) probability measure, say $P^*$. Specifically, $V^i_{t_{k}+\delta}$ is computed as conditional expectation under measure $P^*$ of a future (occurring after time $t_{k}+\delta$) discounted relevant  cash flows. The conditioning is done with respect to the sigma field $\sF_{t_{k}+\delta}$. On the other hand, the risk measure $\gls{rho}_{t_k}$ is based on the statistical measure $P$ and on the information carried by $\sF_{t_k}$.
\begin{remark} \label{rem:DRM}
Several comments are in order.

\smallskip\noindent
(i) Throughout, we take the risk management point of view at the cash flows and the computation of the corresponding risk. Accordingly, for us, a dynamic risk measure is a monotone decreasing, local,  normalized and cash-additive function. We refer the reader to e.g. \cite{BCP2017,AcciaioPenner2010} for a comprehensive survey of the theory of dynamic risk measures.

We want to stress that, in contrast to the convention adopted in the counterparty risk literature, in our setting the risk of a negative cash flow is a positive quantity, and a positive cash flow has no risk (or negative risk).
For example, if $\gls{rho}_{t_k}(-(X^i_{t_k})^+)$ is computed as classical dynamic $\var_\alpha$, then $\rho_{t_k}(-(X^i_{t_k})^+)$ is the lower $\alpha$-quantile of $-(X^i_{t_k})^+$ computed under measure $P$ and conditioned on $\sF_{t_k}$.
This explains the negative sign in  \eqref{eq:IMmethod2}.

\smallskip\noindent
(ii) From the risk evaluation point of view, the CCP's ``risk'' associated with the cash flow $X^i_{t_k}$ is   equal to $\rho_{t_k}(-X^i_{t_k})$, which could be in principle negative (the CCP should make a payment) and lead to inconsistency since CCPs do not pay any IM to the members. This is the reason of taking positive part of the exposure in \eqref{eq:IMmethod2}.

\smallskip\noindent
(iii)
An alternative proposal for computation of the IM called from the $i$th member at time $t_k$ might be
\begin{equation}\label{eq:IMmethod1}
\textrm{IM}^i_{t_k} = \left(\rho_{t_k}(-X^i_{t_k})\right)^+.
\end{equation}
Note that the IM$^i$ computed via \eqref{eq:IMmethod2} is greater  than that computed via \eqref{eq:IMmethod1}.
Indeed,  $X^i_{t_k}\leq (X^i_{t_k})^+$, and thus $-X^i_{t_k}\geq - (X^i_{t_k})^+$, and consequently
$
  \rho_{t_k}(-X^i_{t_k}) \leq \rho_{t_k}(-(X^i_{t_k})^+).
$
   Since $-(X^i_{t_k})^+\leq 0$, we get that $\rho(-(X^i_{t_k})^+)\geq 0$, and by taking positive part of both parts in the inequality above we deduce that
  $$
  (\rho_{t_k}(-X^i_{t_k}))^+ \leq \rho_{t_k}(-(X^i_{t_k})^+).
  $$
  Thus,  IM$^i$ computed via \eqref{eq:IMmethod2} is more conservative, as expected, since CCP does not care about the gains of its members, but only about its losses.

\smallskip\noindent (iv) Both \eqref{eq:IMmethod2} and \eqref{eq:IMmethod1}, in principle can be zero, which means that a CM does not post any IM to the CCP - which, practically,  is a strange situation. However, under normal market conditions, the distribution of $X^i_{t_k}$ will have both positive and negative parts, with lower  quantiles (say 10\%) being negative.
In this case, for a large class of risk measures, such as $\var$ and Expected Shortfall, the IM computed by \eqref{eq:IMmethod2} and \eqref{eq:IMmethod1} will be equal and will coincide with $\rho_{t_k}(-X^i_{t_k})$. However, for other risk measures, such as Entropic Risk Measure\footnote{The Dynamic Entropic Risk measure of a future cash flow $X$ is defined as $
\textrm{Ent}(-X) = \frac{1}{\alpha}\log E\left[ e^{\alpha X}\right],
$
where $\alpha\in(0,1)$. }, this may not be the case.
\end{remark}
A detailed analysis of how the choice of the risk measure $\rho$ impacts the default waterfall will be performed in Section~\ref{sec:numerical}.

\subsection{Prefunded Default Fund}\label{sec:DF}
In the current market practice, the prefunded DF (or simply DF) usually is computed following the Cover1/Cover2 principle, which stipulates that the DF should cover CCP losses in case when one or two CMs with largest exposure to the CCP default. In our opinion, and, in fact, in the opinion of practitioners and regulators, this is not the appropriate principle. One of the reasons for this is that, typically, the computed DF levels cover the losses induced from the potential defaults of two largest CMs do not guarantee the complete coverage of these losses, meaning that potential losses exceeding the DF occur with positive probability. Another possible reason, we think, is that in the current practice of CCP computation of the DF is done assuming that the default times of the CMs are independent. Moreover, usually it is assumed that the credit worthiness of all members is the same. These drawbacks are addressed in the proposed model. We refer to Section~\ref{sec:numerical} for the analysis of how our model performs vis-a-vis these drawbacks.

We start by introducing the  liquidation or auction value of the CM$^i$'s portfolio. If the CM$^i$ defaults at time $\gls{taui}=t_k$, then the outstanding CM$^i$'s portfolio is either liquidated or auctioned over the period of time $\delta$.
We denote by $\gls{hatVi}_{t_{k}+\delta}$ the value of this outstanding portfolio that is recovered by time $t_{k}+\delta$.
 Consequently, we define the \textit{net exposure} at time $T_k$ of the CCP to the default of the $\gls{CMi}$ net the IM and VM as
\begin{align}\label{eq:preexp}
\gls{EPi}_{T_k} =  \beta_{T_k}^{-1}\sum_{t_m=T_k+\delta_f}^{T_{k+1}}
\Big( \beta_{t_m+\delta} & [V^i_{t_m+\delta}- \widehat V^i_{t_{m}+\delta}]+  \sum_{u=t_{m}}^{ t_{m}+\delta}\beta_{u}D^i_u \nonumber \\
& \qquad - \beta_{t_m} [\textrm{VM}^i_{t_m} + \textrm{IM}^i_{t_m}] \Big)^+\1_{\tau_i=t_m}.
\end{align}
Without much loss of generality, we assume that for each $T_k$ the random variable $\textrm{EP}^i_{T_k}$ is bounded (this can always be achieved by capping the exposures by a large constant).

\begin{remark} (i)
Calculation of the liquidation or auction value of CM$^i$'s portfolio, that is, calculation of the recovery term $\gls{hatVi}_{t_{k}+\delta}$, is an important  aspect of computation of the net exposure. We will briefly mention two possible approaches.

A Brazilian CCP, called BM\&FBOVESPA, has put effort into developing a methodology, named CORE, aimed at optimal liquidation of a multi-asset portfolio and thus generating an optimal (``close to the marketed'') liquidation value of CM$^i$'s portfolio (cf. e.g. \cite{CORE-2015}).

On the other hand, guided by common practice adopted in the banking industry regarding asset recovery valuation, one may propose to compute the liquidation or auction value of CM$^i$'s portfolio as a fixed fraction of the marketed value $V^i_{t_m+\delta}$ of this portfolio, so that, say, $\gls{hatVi}_{t_{k}+\delta}:=R V^i_{t_m+\delta}$, where $R$ is the recovery rate. We will adopt this convention in our numerical stress tests (assuming also that $R$ is a deterministic constant).

(ii) In the current regulatory practice  the net exposure does not account for the liquidation or auction value of CM$^i$'s portfolio to the effect that the \textit{regulatory net exposure} is taken to be
\begin{equation}\label{eq:reg-preexp}
\gls{EPiReg}_{T_k}= \beta_{T_k}^{-1}\sum_{t_m=T_k+\delta_f}^{T_{k+1}}
\left ( \beta_{t_m+\delta} V^i_{t_m+\delta} - \beta_{t_m} (\textrm{VM}^i_{t_m} + \textrm{IM}^i_{t_m}) \right )^+\1_{\tau_i=t_m}.
\end{equation}
This of course is, in general, more conservative evaluation of the net exposure. It really amounts to postulating a ``no--recovery'' paradigm.
\qed
\end{remark}

Similarly to the computation of IM, we propose to use a dynamic risk measure, say $\gls{eta}$, for computing the total DF.
The proper choice of $\eta$ is an important issue, and it is discussed in the next section.
The total DF, at time $T_k$, is defined as
\begin{equation}\label{eq:DF}
\textrm{DF}_{T_k} = {\eta_{T_k}} \left(-\sum_{i\in \cI}\textrm{EP}^i_{T_k} \right).
\end{equation}
Note that \gls{EPi} is a non-negative quantity, and thus $\textrm{DF}_{T_k}$ is always non-negative too, given that $\eta$ is a normalized risk measure. Similarly as in the case of the initial margin,  the risk measure $\eta_{T_k}$ is based on  statistical measure $P$ and on the information carried by $\sF_{T_k}$.

\subsubsection{DF allocation}\label{sec:DFAllocation}
Once the size of the DF is established, the CCP has to allocate it among CMs in some fair way.
Usually, in the current practice, the allocation is done proportionally to the initial margin of each CM. We will propose here an alternative approach.

Towards this end, we first note that from the mathematical point of view DF allocation is related to the so called capital allocation, and, just as in the case of capital allocation, one runs into various difficulties.  Generally speaking, the capital allocation based on risk contribution is not an easy issue to deal with. There exists significant research devoted to this topic, in particular by using risk measures as the main tool; for static risk measures see for instance \cite{Tasche2000,Denault2001,cherny-2006p2,Delbaen2000,Fischer2003,Kalkbrener2005,Tasche2002,ArmentEtAl2015,BrunnermeierCheridito2014}, and for dynamic risk measures see for instance \cite{Cherny2009a,KromerOverbeckZilch2015a}.

One of the DF allocation principles that has been suggested in the literature, is the Euler principle (cf. \cite{Tasch2007,EmbrechtsLiuyWangz2016}). Application of this principle requires that certain technical conditions are satisfied, which  are not satisfied in our model. Also, the Euler principle is not popular among practitioners. Here we propose a DF allocation principle, which we think will be an adequate tool to be used by CCPs.

In order to suggest a DF allocation scheme, we proceed by assuming that $\eta$ is a dynamic coherent  risk measure \cite{ArtznerDelbaenEberHeathKu2002b}.
Then, $\eta$ admits a robust representation of the form
\begin{equation}\label{eq:robust}
  \eta_t(X) = \esssup_{Q\in\cQ_t}E^Q[-X\mid\sF_t],
\end{equation}
where $\cQ_t$ is a set of probability measures absolutely continuous with respect to $P$, such that for any $Q\in\cQ_t$, $Q=P$ on $\sF_t$, and that satisfies some additional technical properties.
Under some mild assumptions on $\cQ_t$, the essential infimum in \eqref{eq:robust} is attained, say at $Q^*_t$.
We denote by $Z^*_t$ the Radon-Nikodym  derivative  $dQ_t^*/dP$.
Consequently, for $t=T_k$ and $X=-\sum_{i\in \cI}\textrm{EP}^i_{T_k}$, we have
\[
\eta_{T_k}\left(-\sum_{i\in \cI}\textrm{EP}^i_{T_k} \right) = E[Z^*_{T_k}\sum_{i\in \cI}\textrm{EP}^i_{T_k}|\sF_{T_k}],
\]
and so
\[
\textrm{DF}_{T_k} = \sum_{i\in \cI}E[Z^*_{T_k}\textrm{EP}^i_{T_k}|\sF_{T_k}].
\]
We propose to define the individual default fund contributions $\gls{DFi}_{T_k}$ as
\begin{equation}\label{eq:DFAlloc}
\gls{DFi}_{T_k}:=E[Z^*_{T_k}\gls{EPi}_{T_k}|\sF_{T_k}].
\end{equation}
In particular, this leads to an important consistency property of the proposed DF allocation
\begin{equation}\label{eq;sumDF}
\sum _{i\in \mathcal{I}}\gls{DFi}_{T_k}= \textrm{DF}_{T_k}.
\end{equation}

In addition, the default fund allocation done according to \eqref{eq:DFAlloc}, enjoys another key property: it provides the monotonicity of the allocation with respect to the net exposures. In the extreme case, if $\gls{EPi}_{T_k}\geq \textrm{EP}^j_{T_k}$ then $\gls{DFi}_{T_k}\geq \textrm{DF}^j_{T_k}$.

\begin{remark}\label{remark:nonUnique}
It is important to observe that the maximizers $Z^*_{T_k}$ are not necessarily unique. This, for instance, is the case if the  default fund allocation rule is based on the conditional average value at risk, as discussed in the example   below. Consequently, the individual default fund contributions $\gls{DFi}_{T_k}$ defined as in \eqref{eq:DFAlloc} are not unique, in general. This gives the CCP flexibility in allocating the default fund to individual members.
\end{remark}

\begin{example}
A special example of the default fund allocation rule, that we will use in our numerical studies presented in Section~\ref{sec:numerical}, is based on $\eta$ given as the Conditional Average Value at Risk (\gls{AVaR}), also known as the conditional expected shortfall.\footnote{If the distribution of a random variable is continuous, then its conditional expected shortfall coincides with conditional  tail conditional expectation.}
This measure is attractive both from the mathematical perspective and from the practical perspective.
It is a dynamic coherent utility measure (see e.g.~\cite{ChernyWVAR2006}), and it is supermartingale time-consistent (see e.g.~\cite{BCP2014a}).

As it is shown in the Appendix~\ref{sec:avar}, there exists a maximizer  $Q^*$ in \eqref{eq:robust}, with the corresponding Radon-Nikodym density $Z^*$
$$
Z^* = \frac{1}{\alpha} (\1_{X<q^\pm_\alpha(X\mid \sF_t)} + \varepsilon \1_{X=q^\pm_\alpha(X\mid \sF_t)} ),
$$
where $q^\pm_\alpha(X\mid \sF_t)$ is the conditional upper/lower $\alpha$-quantile of $X$, and
$$
\varepsilon =
\begin{cases}
  0, & \mbox{if } P(X=q^\pm_\alpha(X\mid \sF_t))=0 \\
  \frac{\alpha-P(X<q^\pm_\alpha(X\mid \sF_t))}{P(X=q^\pm_\alpha(X\mid \sF_t))}, & \mbox{otherwise}.
\end{cases}
$$

\end{example}

\begin{remark}
It turns out that the proposed allocation scheme with $\eta_t(\,\cdot\,)=\avar_\alpha(\,\cdot\,\mid\sF_t)$ coincides with the corresponding Euler allocation scheme for conditional $\avar$ \cite{Tasch2007}. However, generally speaking these two schemes are different.
\end{remark}

\subsection{CCP Skin-in-the-Game and Assessment Power}
In order to complete the model for the unfunded default fund we need to describe the way in which the CCP equity process (or skin-in-the-game) is formed. There is no consensus regarding rules for formation of this process, and usually it represents a percentage (e.g. 20-25\%) of the CCP regulatory capital.  This part of the default waterfall is  usually small in comparison to other parts of the  waterfall.

%\subsection{Unfunded Default Fund (\gls{uDF})}
For $t_k\in[T_j,T_{j+1})$, we  set the effective CCP's loss at time $t_{k}+\delta$ to be given as
\begin{align*}
\textrm{EL}_{t_{k}+\delta}
& = \beta^{-1}_{t_k+\delta}\Big( \sum _{i\in \mathcal{I}}\left(\beta_{t_{k}+\delta}V^i_{t_{k}+\delta} - \beta_{t_k+\delta}\widehat V^i_{t_{k}+\delta}
-\beta_{t_k}\textrm{VM}^i_{t_k} - \beta_{t_k}\textrm{IM}^i_{t_k}\right )\1_{\tau_i=t_k} \\
 & \qquad -  \beta_{t_k+\delta}{\gls{SG}_{t_{k}+\delta}}- \beta_{T_j}\textrm{DF}_{T_j}\Big)^+.
\end{align*}
Then, the \gls{uDFi} for the $\gls{CMi}$ at time $t_{k}+\delta$, which is essentially the last resort before the CCP defaults (see the default waterfall diagram), is defined as
\begin{equation}\label{unfund1}
\gls{uDFi}_{t_{k}+\delta} = \frac{\textrm{DF}^i_{T_j}\1_{\tau_i>T_j}}{\sum _{l\in \mathcal{I}} \textrm{DF}^l_{T_j}\1_{\tau_l>T_j}}
\textrm{EL}_{t_{k}+\delta},
\end{equation}
where by convention $\frac{0}{0}=0.$

%%%%%%%%%%%%%%%%%%%%%%%%%%%%%%%%%%%%%%%%%%%%%%%%%%%%%%%%%%%%%%%%%%%%%%%%%%%%%%%%%%%%%%%%%%%
\section{Numerical Studies}\label{sec:numerical}

We present in this section the simulation results for a CCP model consisting of 8 CMs, each holding a portfolio of 4 single name CDS contracts. Our major objective is to compute the DF/IM ratio. There are at least two reasons for that. First, the magnitude of this ratio indicates whether the CCP's DF risk management practices provide adequate risk reserve. Typically, in the CCP practice, the DF/IM ratios take values around 0.1, which can serve as benchmark for testing one aspect of the model performance. The numerical study that we have done indicates that our model performs satisfactorily in this regard. Secondly, as postulated by the regulators,  any prudent way of computing DF and IM should keep the DF/IM ratio invariant with respect to the number of CMs. Our numerical study indicates that our model performs satisfactorily in this regard as well; we stress tested the  DF/IM ratio different numbers if CMs and the ratio is pretty much constant.

It needs to be emphasized here that the predominantly used in practice cover one and cover two (C1 and C2 for short) methodologies for computing the DF suffer from inadequate scaling of the DF with the changing number of CMs. In fact, the DF computed according to these methodologies is essentially invariant with regard to the increasing number of CMs, as long as the CMs are all alike as far as their portfolio composition is concerned. On the other hand, the IM clearly increases with the number of CMs. So, the resulting DF/IM ratio decreases with the increasing number of CMs, which is an unreasonable and highly undesired feature of course. So, replacing the predominant C1/C2 methodologies with a method like ours, may significantly contribute to improvement of robustness of CCP operations.

In the process of computing the DF/IM ration, using our model, we stress test these two components of the default waterfall. We ignore the other elements of the default waterfall since, as we believe, they are less relevant for the premise of this paper as stated in the Introduction.

Regarding the simulation part of our computations, we note that DF is computed based on IM, which leads to a nested simulation in numerical studies. Because of this we assume constant default intensities for all CDSs in our portfolios, in which case IM can be computed explicitly. As a result, we rule out the nested simulation error in DF computation and we reach a more reliable DF$^i$/IM$^i$ ratio.

Although we focus in our numerical study on a stylized model of 8 CMs and portfolios of CDS contracts, it needs to be stressed that, in principle,  the proposed  theory can be applied to  any number of CMs and to portfolios of any traded contracts.

As said above, for each CM we consider a portfolio consisting of 4 single name CDS contracts. We assume  that at the times at which IM and DF are computed all the listed CDS contracts in the portfolios are alive; otherwise they will not be listed. We adopt the post-Big Bang convention for mechanics of the CDS contracts. We refer to  Appendix \ref{sec:bigbang} for more details and formulae that we use to compute MtM for CDSs.

Throughout this section, we assume that there are 252 business days in one year, and we take the fundamental time $\delta_f$ to be one business day, the margin period of risk $\delta$ to be 10  business days, and $\Delta_f$ to be 30 business days. All relevant spreads, and rates are given on annualized basis. We fix the recovery rate as $R=0.4$, and the CDS spread $\kappa$ as 0.01 (100 bps). We maintain 4 (alive) CDS contracts. The constant default intensities underlying the contracts are listed in the presentation of the numerical results below.  For simplicity, we assume these contracts to be all initiated on June 20, 2015 and to  mature on June 20, 2018. The results we show are computed for Sept 22, 2015.

Let $H=(H_{ij})_{8\times 4}$, where $H_{ij}$ represents the position held by the $i$-th CM in the $j$-th CDS contract, $i=1,\ldots,8,j=1,\ldots,4$, as seen from the CCP perspective. Specifically, the positive $H_{ij}$  means that the CCP has  long position (buys protection) in $H_{ij}$ shares of the $j$-th CDS contract relative to the $i$-th CM; analogously the negative  $H_{ij}$  means that the CCP has  short position (sells protection). Since the CCP runs a matched order book, we have that $\sum_{i\in \mathcal{I}} H_{ij}=0$. We will use the notation $H^i=(H_{i1},\ldots,H_{i4})$ to denote the positions of CCP with respect to CM$^i$.

\subsection{Initial Margin}\label{sec:NumIM}
Even though our model is set in discrete time, we use the simplifying continuous time convention for computing the mark-to-market value of the $j$-th CDS contract at time $t_k+\delta$, which we denote by $S_{t_k+\delta}^j$, and which is used for  computing $V^i_{t_k+\delta}$ in formula \eqref{eq:Xitk}. In particular, in the valuation of CDS contracts, the CDS spread is assumed to be paid continuously; see Section \ref{sec:bigbang}. %Since we take zero interest rates in the simulation, this assumption makes no difference from lump premium payment.
On the contrary, when computing CCP's exposure at time $t_k$, resulting from the $j$-th CDS, we consider lump CDS spread (coupon) payment. Thus, given our assumption of zero interest rates,  the  CCP exposure at time $t_k$, resulting from the $j$-th CDS, is given as $\check S_{t_k}^j:=S_{t_k+\delta}^j+\sum_{u=t_k}^{t_k+\delta}d_u^j-S_{t_k-1}^j,$
where $d_u^j$ is the  dividend associated with the $j$-th CDS at time $t_k$. Thus
(cf. \eqref{eq:Xitk}),
\begin{equation}\label{eq:exp} X_{t_k}^i=\sum_{j=1}^{4} H_{ji} \check S_{t_k}^j.\end{equation}
In our set up there is either one or none CDS spread payment at dates between $t_k$ and $t_k+\delta$. Assuming that the next coupon payment (premium) day is $T_D > t_k$, and the last premium coupon payment occurred on $t_D\leq t_k$, we see that the term $\sum_{u=t_k}^{t_k+\delta}d_u^j$ takes the form
\begin{equation}\label{eq:div}	
	\begin{cases}
		-\kappa (T_D-t_D)\1_{t_k<T_D\le t_k+\delta},\quad \phi^j>t_k+\delta,\\
		-\kappa( \phi^j-t_D)-S_{t_k-1}^j, \quad t_k<\phi^j \le t_k+\delta,
\end{cases}
\end{equation}
Accordingly, the exposure $\check S_{t_k}^j$ can be written as
\begin{equation}\label{eq:loss}
	\check S_{t_k}^{j}=
	\begin{cases}
		S^j_{t_k+\delta}-\kappa (T_D-t_D)\1_{t_k<T_D\le t_k+\delta}-S_{t_k-1}^j,\quad \phi^j>t_k+\delta,\\
		R-\kappa( \phi^j-t_D)-S_{t_k-1}^j\approx R- \kappa (t_k-t_D)-S_{t_k-1}^j, \quad t_k<\phi^j \le t_k+\delta.
\end{cases}
\end{equation}

Putting $p_j:=P(\phi^j>t_k+\delta \mid \phi^j>t_k)$, we obtain
\begin{align}\label{eq:dist}
  \begin{aligned}
&		P(\check S_{t_k}^j=S^j_{t_k+\delta}-\kappa (T_D-t_D)\1_{t_k<T_D\le t_k+\delta}-S_{t_k-1}^j \mid \phi^j>t_k)=p_j,\\
&		P(\check S_{t_k}^j= R- \kappa (t_k-t_D)-S_{t_k-1}^j\mid\phi^j>t_k)=P(\phi^j\le t_k+\delta \mid \phi^j>t_k)=1-p_j.
\end{aligned}
\end{align}
In this study  we assume constant intensity of the default times of the reference names underlying the CDS contracts, in which case the conditional distribution of $X_{t_k}^i$ given $\sF_{t_k}$ can be computed explicitly; see Section~\ref{sec:bigbang} for details.
Also note, that $X_{t_k}^i$ takes only finitely many values, say $\{x_{1}^i,x_{2}^i,\ldots,x_{N}^i \}$, for some $N\leq 16$, and
without loss of generality, we assume that $x_{n}^i\geq x_{n+1}^i$, for all $n=1,\ldots,N-1$.
In this case, the  conditional $\var_\alpha$ and $\avar_\alpha$ for a fixed confidence level $\alpha$ are computed as
\begin{align*}
\var_\alpha(-(X_{t_k}^i)^+\mid \sF_{t_k}) & =(x_{n_\alpha}^{i})^+, \\
\avar_\alpha(-(X_{t_k}^i)^+\mid \sF_{t_k}) & =
\frac{1}{\alpha} ((x_{1}^i)^+ p^i_1+\ldots+(x_{n_\alpha}^i)^+ (\alpha-\sum_{n=1}^{n_\alpha} p_n^i)),
\end{align*}
where  $p^i_n = P(-(X_{t_k}^i)^+=-(x^i_{n })^+ \mid \sF_{t_k})$, $n=1,\ldots,N$, and $n_\alpha=\min\{k \mid \sum_{n=1}^{k} p_n^i>\alpha\}$.

Towards this end, we take the default intensities for the four considered CDSs to be respectively equal  to
$$
\lambda_1=0.002, \quad \lambda_2=0.01, \quad \lambda_3=0.015, \quad \lambda_4=0.03.
$$

It needs to be stressed that, in what follows, the numerical values of $\var_\alpha(-(X_{t_k}^i)^+\mid \sF_{t_k})$ and $\avar_\alpha(-(X_{t_k}^i)^+\mid \sF_{t_k})$ are computed on the event that at time $t_k$ all the CDS contracts and all the CMs are still alive.

We start by computing the numerical values of the CCP exposure for holding a long position in each CDS contract.
In Table~\ref{table:dist} we present the values of $S_{t_k}^{j}$ assuming that there is no premium due during the margin period of risk $\delta$.
{\small
\begin{table}[H]
\caption{Conditional Distribution of $\check S_{t_k}^{j}$. }
	\begin{center}
		\begin{tabular}{|r |c |c| c|c|c|}
\hline
			& \quad CDS$^1$ \quad &  \quad  CDS$^2$  \quad  & \quad   CDS$^3$ \quad  & \quad  CDS$^4$ \quad  \\	
\hline		
			$\check S_{t_k}^{j} \mid \phi^j>t_k+\delta $     &0.0004 &0.0003 &0.0002 &-0.00008\\
			$p_j$                                            &0.9999 &0.9996 &0.9994 &0.9988\\
\hline
$\check S_{t_k}^{j} \mid t_k<\phi^j \le t_k+\delta $        &0.4252 &0.4162 &0.4107 &0.39\\
			$1-p_j$                                          &0.0001 &0.0004 &0.0006 &0.0012\\
\hline
		\end{tabular}
	\end{center}
\label{table:dist}
\end{table}}
As one may expect the (absolute) value of the exposure from individual CDS contracts is large if the reference entity  defaults, although this happens with small probability. Also note that a short position (CCP sells protection) in a CDS contract does not contribute significantly to the CCP's exposure towards the corresponding member. Hence, short positions do not net out significantly the exposure from the long position with the same CM.
Consequently,  the distribution of $X_{t_k}^i$ will be skewed, and skewness being determined by the overall number of long positions of CCP with the CM$^i$. A larger number of independent CDS contracts in the portfolio will increase the skewness too.

Since the value of the IM depends only on the individual CM position, for the rest of this section, let us consider three CMs portfolios that will be analyzed separately
$$
H^1 = (10,10,-1,-1), \quad  H^2 = (10,-100,5,5), \quad  H^3 = (1,-100,-100,-100).
$$
The first and, respectively  the second, portfolio, is moderately long, and respectively short,  in aggregate.
The third portfolio is extremely short in aggregate. In Table~\ref{table:dist1}-\ref{table:dist3},  using \eqref{eq:exp}, \eqref{eq:loss}, \eqref{eq:CDS} and \eqref{eq:CDSpre}, we present the probability distribution function (pdf) of the CCP exposure  $-(X_{t_k}^i)^+$ for each of the above portfolios.
Note that $H^1$ and $H^2$ have the same number of long positions, while $H^2$ has significantly larger short position.
On the other hand, $H^1$ and $H^2$ have similar maximal exposure around $8.30$. The  probability distribution of the exposures corresponding to $H^1$ and $H^2$ is similar, and thus we will compute the IM only for $H^1$. Also note that $H^3$ exhibits  an almost zero exposure, and thus the IM will be  zero if computed using $\var$, and close to zero if one uses  $\avar$.

\begin{table}[H]
	\caption{The pdf of $-(X_{t_k}^1)^+$}\label{table:dist1}
	\begin{center}
		\begin{tabular}{|c |c |c| c|c|c|}
			\hline
			$-(X_{t_k}^1)^+$ &$-8.41$                &$-8.02$                &$-8.00$                &$-7.61$                &$-4.25$  \\
			\hline
			$p$              &$3.2\times 10^{-8}$    &$3.8\times 10^{-11}$   &$1.9\times 10^{-11}$   &$2.2\times 10^{-14}$   &$7.9\times 10^{-5}$\\
			\hline\hline
			$-(X_{t_k}^1)^+$ &$4.17$                 &$-3.86$                &-3.84                  &-3.77                  &-3.76   \\
			\hline
			$p$              & $4.0\times 10^{-4}$   &$9.5\times 10^{-8}$    & $4.7\times 10^{-8}$   &$4.7\times 10^{-7}$    &$2.4\times 10^{-7}$ \\
			\hline\hline
            $-(X_{t_k}^1)^+$ &-3.45                  &-3.36                 &-0.0065                 &0                     & \\
\hline
            $p$              &$5.6\times 10^{-11}$   &$2.8\times 10^{-10}$   &$0.998$               &0.0018                 &\\
\hline
		\end{tabular}
	\end{center}
\end{table}

\begin{table}[H]
\caption{The pdf of $-(X_{t_k}^2)^+$}\label{table:dist2}
	\begin{center}
		\begin{tabular}{|r |c |c|c|c|c|}
			\hline
			$-(X_{t_k}^2)^+$ &-8.25                  &-6.28            &-6.20                        &-4.223\\
			\hline
			$p$              &$5.6\times 10^{-11}$   &$4.7\times 10^{-8}$    &$9.5\times 10^{-8}$    &$7.9\times 10^{-5}$\\
			\hline
			$-(X_{t_k}^2)^+$ &-4.01                  &-2.03                  &-1.95                  &0\\
			\hline
			$p$              &$7.1\times 10^{-7}$    &$6.0\times 10^{-4}$    &$1.2\times 10^{-3}$    &0.998\\
			\hline		
		\end{tabular}
	\end{center}
\end{table}

\begin{table}[H]
	\caption{The pdf of $-(X_{t_k}^3)^+$}\label{table:dist3}
	\begin{center}
		\begin{tabular}{|r |c |c|}
			\hline
			$-(X_{t_k}^3)^+$&$-0.39$&$0$\\
			\hline
			$p$&$7.93\times 10^{-5}$&0.999921\\
			\hline
		\end{tabular}
	\end{center}
\end{table}

Next, we will analyze the IM computed using $\var$ and $\avar$ for $H^1$.
Figure~\ref{fig:var1}, left panel  shows that $\var$ is not sensitive to usually applied industry risk level $\alpha = 0.01$ or $\alpha=0.05$.
Specifically,  $\var_\alpha(-(X_{t_k}^1)+)  = 0.0065 $, for $ \alpha \in [0.00477, 0.998)$.
Thus, the numerical results indicate that $\var$ is not an appropriate risk measure for computing the IM. This observation is consistent with the industry practice of using $\avar$ instead for $\var$ in risk management. The values of $\avar$ for $H^1$, given in Figure~\ref{fig:var1} right panel, indicate that $\avar$ indeed is better suited for risk management applications at CCPs.

Note that $\textrm{IM}^1 = \avar_{0.01}(-(X_{t_k}^1)^+) = 0.21$, which is roughly 1\% of 20 - the total notional of long positions in $H^1$.
These results agree with industry practice where reasonable IM is 1-2\% of the notional, with $\alpha=0.01$ as generally accepted risk level.
%The results also confirm that only the long positions contribute to the IM.

\subsection{Default Fund}\label{sec:NumDF}
The default fund is computed based on the exposure $\gls{EPi}_{T_k}$ defined in \eqref{eq:preexp}.
We assume that the default times of the underlying CDS contracts are independent of the CMs credit migration processes.
We compute the DF by Monte Carlo method.
The numerical  experiments indicate that to achieve a reasonable accuracy and to balance the computational cost,
it is enough to simulate 100 paths for each $\phi^j$ and 10000 paths for $\tau^i$.

Since we assume constant default intensity for the CDS contracts, and thus, the simulation of $\phi_j$ is reduced to drawing  i.i.d. unit exponential random variables $\mathcal{E}_j$, and take $\phi_j=\inf \{ t>0: \lambda_j t\ge \mathcal{E}_j\}$.

In this section, we will consider the following two portfolios hold by CCP,
$$
H_b=
\begin{bmatrix}
1&-1&1&-1\\
-1&1&-1&1\\
10&-1&-8&-1\\
-1&2&-2&1\\
-10&5&-5&10\\
-1&-1&-5&7\\
20&10&18&-48\\
-18&-15&2&31.
\end{bmatrix}, \qquad
H_u=
\begin{bmatrix}
1&1&1&1\\
10&-1&10&-1\\
-1&10&-1&10\\
100&-5&100&-5\\
-110&-5&-110&-5\\
-1&-1&-1&-1\\
-2&-1&-6&-3\\
3&2&7&4\\
\end{bmatrix}.
$$

Portfolio $H_b$ is mostly uniform across the CM's positions and we will call it balanced, while $H_u$ will be referred as an unbalanced portfolio (with two `large' CMs).

To simulate the default times for all the CMs, we adopt a discrete time homogeneous (strong)
Markovian structure model for credit migrations.  This in particular means that the credit migration processes for all the CMs are time homogeneous Markov chains, and that the joint migration process is also a time homogeneous Markov chain. See Appendix~\ref{sec:MC} for details on Markovian structures.

We assume 8 credit rating levels $\set{1,2,\ldots,7,8}$, with  1 being the highest credit rating, and 8 being the default.
Transition to the default level 8 may only occur from levels  $\set{3,\ldots,7}$; we call transitions from levels $\set{3,\ldots,6}$ to the default level 8 \textit{the jump to default}. For simplicity we assume that the ratings can only change between neighboring values, unless the CM jumps to default. Finally, we  assume that the marginal (or individual) migration processes are governed by the same  transition matrix for all CMs.

We consider three types of dependence structure between credit levels of the CMs.
\begin{enumerate}[]
  \item $\quad $ Type I, in which the marginal (individual) migration processes of all members are all independent of each other.
  \item $\quad $ Type II, in which any member's jump to default prohibits the other members' credit upgrade.

  \item $\quad $ Type III, in which either credit ratings of all members simultaneously migrate in the same way, or the credit rating of one member migrates and the credit ratings  of all the remaining members stay put.
\end{enumerate}
We  consider only two cases for the CMs' credit ratings at the current time $t_k$: either all the CMs have the highest credit rating of 1, or all the CMs  have the credit rating of 7. We will denote these two cases as $\{\mathbf{1}\}$ and $\{\mathbf{7}\}$.

The individual $\gls{DFi}$s and the total DF are primarily determined by the size of the long positions, which determine the exposure, regardless of whether the portfolio is balanced or not. This confirmed by our simulation results. Accordingly, we will only present the numerical results for the balanced portfolio.

One the major goals of this part of our numerical study is to examine the impact that various values of the risk levels $\alpha_{\textrm{IM}}$ and  $\beta_{\textrm{DF}}$ have on the sizes of the \gls{DFi}s and thus on the size of the total DF.  A sample of the results is presented in Figures \ref{fig:barBTypeII} and \ref{fig:DFratio}. The former one shows the DF/IM ratio for various values of $\alpha_{\textrm{IM}}$ and  $\beta_{\textrm{DF}}$ and for two different initial configurations of credit ratings, and the latter one displays the DF/IM curves for the representative value of $\alpha_{\textrm{IM}}=0.01$, as the function of $\beta_{\textrm{DF}}$.

In Figure~\ref{fig:barBTypeII} we consider only Type~II dependence structure, since analogous results hold for Type~I and Type~III dependence structure. This observation is further confirmed in the graphs shown in Figure \ref{fig:DFratio}.

Figure~\ref{fig:barBTypeII} nicely illustrates the intuitive feature that there is no monotonicity in the way that the  DF/IM ratio depends on $\alpha_{\textrm{IM}}$. It also illustrates another intuitive feature that this ratio decreases in $\beta_{\textrm{DF}}$.

In order to interpret the results in Figure~\ref{fig:DFratio}, we first note that, as expected, there is a significant difference in DF/IM (and thus in total DF amount) between the two initial configurations of credit ratings. For $\beta_{\textrm{DF}} = 0.01$, and initial state $\{\mathbf{1}\}$ we have type III DF/IM =0.0026, while for initial state $\{\mathbf{7}\}$, we have  type III DF/IM = 0.5312. In general the DF/IM ratio for the initial configuration $\{\mathbf{7}\}$ is more than 200 times larger than for $\{\mathbf{1}\}$. We also note that the value DF/IM=10\%, which is often reported by CCPs, is roughly achieved for  $\beta_{\textrm{DF}}=0.05$ as of initial state $\{\mathbf{7}\}$. However, this is achieved for any $\beta_{\textrm{DF}}$ when starting with $\{\mathbf{1}\}$.

Our simulation results also show that, as expected, the \gls{DFi}s increase with the size of the long positions.

As it has been already mentioned in Section~\ref{sec:D-CCP}, so called cover values are computed by CCPs. Using our simulated results for the values of DF and $\gls{DFi}$s,  we computed the empirical probabilities that these default funds cover the exposures. Specifically, we computed  the empirical probability that DF covers the exposure of the CM with the largest exposure, the empirical probability that DF covers the exposure of two CMs with the largest exposure, and the empirical probability that DF covers the exposure of all the CMs. We also computed the empirical probabilities of, what we term -- self-covers, that is the empirical probabilities that the exposure of the CM with the largest exposure is covered by this member's individual default fund (Self-Cover 1), and, similarly for Self-Cover 2. The results are reported in Figures \ref{fig:coverRatio}. The results are only presented for the initial configuration of $\{\mathbf{7}\}$; for the initial configuration of $\{\mathbf{1}\}$ these empirical probabilities are almost equal to 1.

As is shown in Figure \ref{fig:DFratioDiffMem}, DF/IM ratio does not vary much with the number of members or the size of the positions holding by the members. This indicates that our model scales appropriately the sizes of both the DF and IM with the increasing number of CMs. On the contrary, as already said in the beginning of this section, the C1/C2 methodology does not scale appropriately the size of DF with the increasing number of CMs: essentially, the C1/C2 approaches significantly underestimates the size of DF.

All in all, our model proves to be very flexible and capable of producing results that are in agreement the expected features of default waterfall.

\begin{figure}[!ht]
    \centering
\begin{subfigure}[t]{0.495\textwidth}
        \centering
        \includegraphics[width=\linewidth]{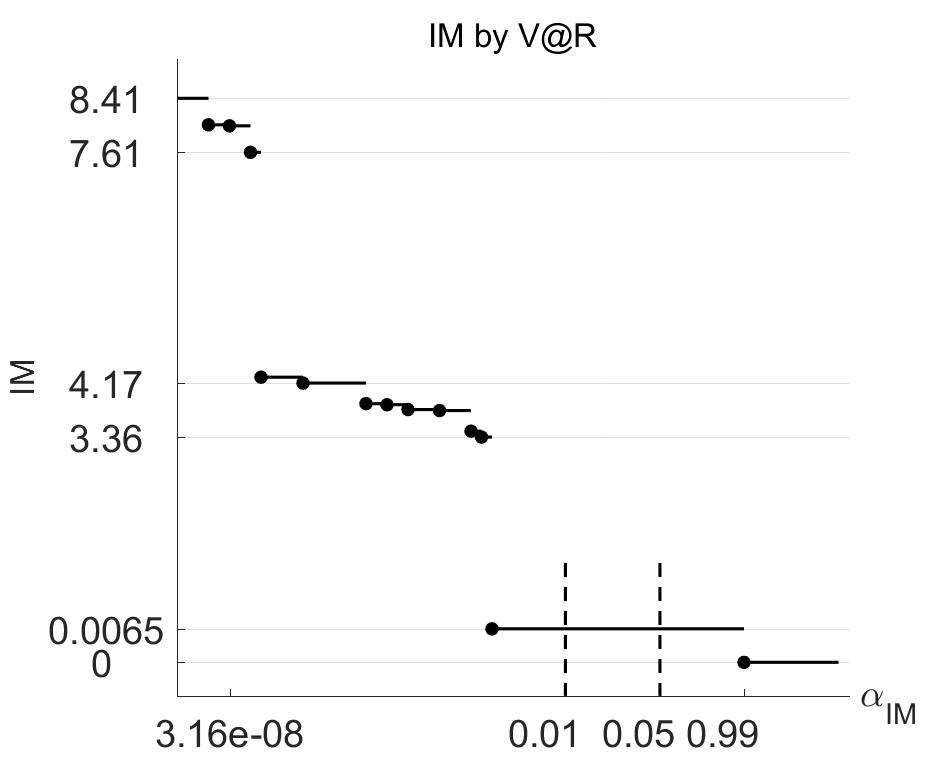}	
\end{subfigure}
\hfill
\begin{subfigure}[t]{0.495\textwidth}
        \centering
        \includegraphics[width=\linewidth]{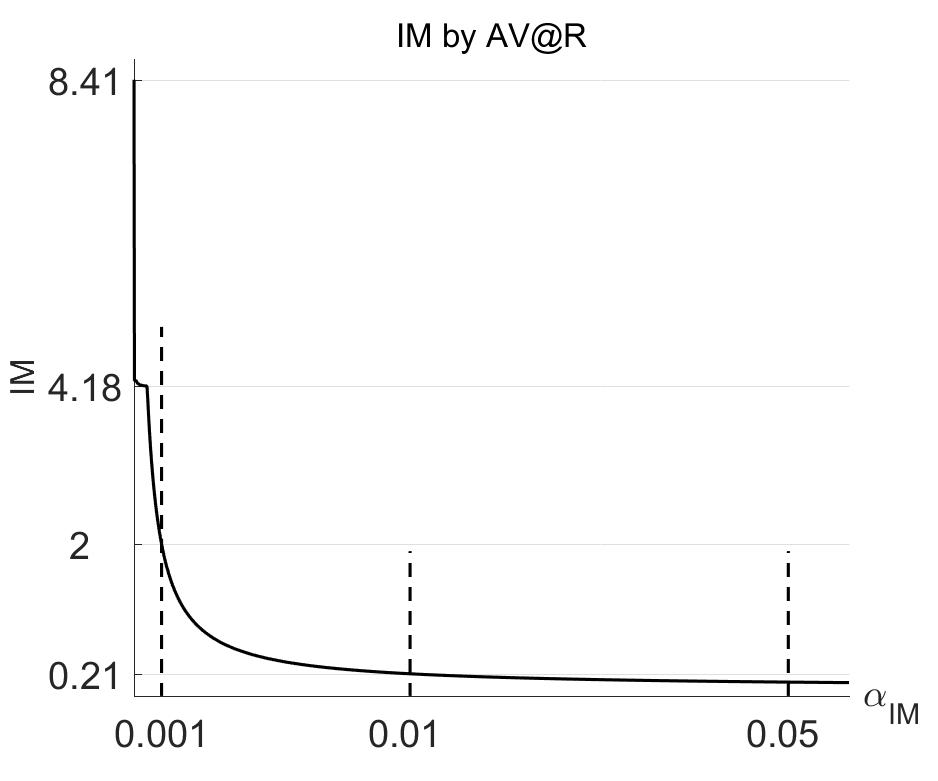}
\end{subfigure}
\caption{IM for $H^1$ computed by $\var$ and $\avar$}\label{fig:var1}
\end{figure}

%%%%%% this block does not compile in arxiv
\begin{figure}%[!ht]%\caption{DF/IM with Various $\alpha_{IM}$ and $\beta_{DF}$}
  	\centering
    \begin{subfigure}[t]{0.495\textwidth}
		\centering
\includegraphics[width=\linewidth]{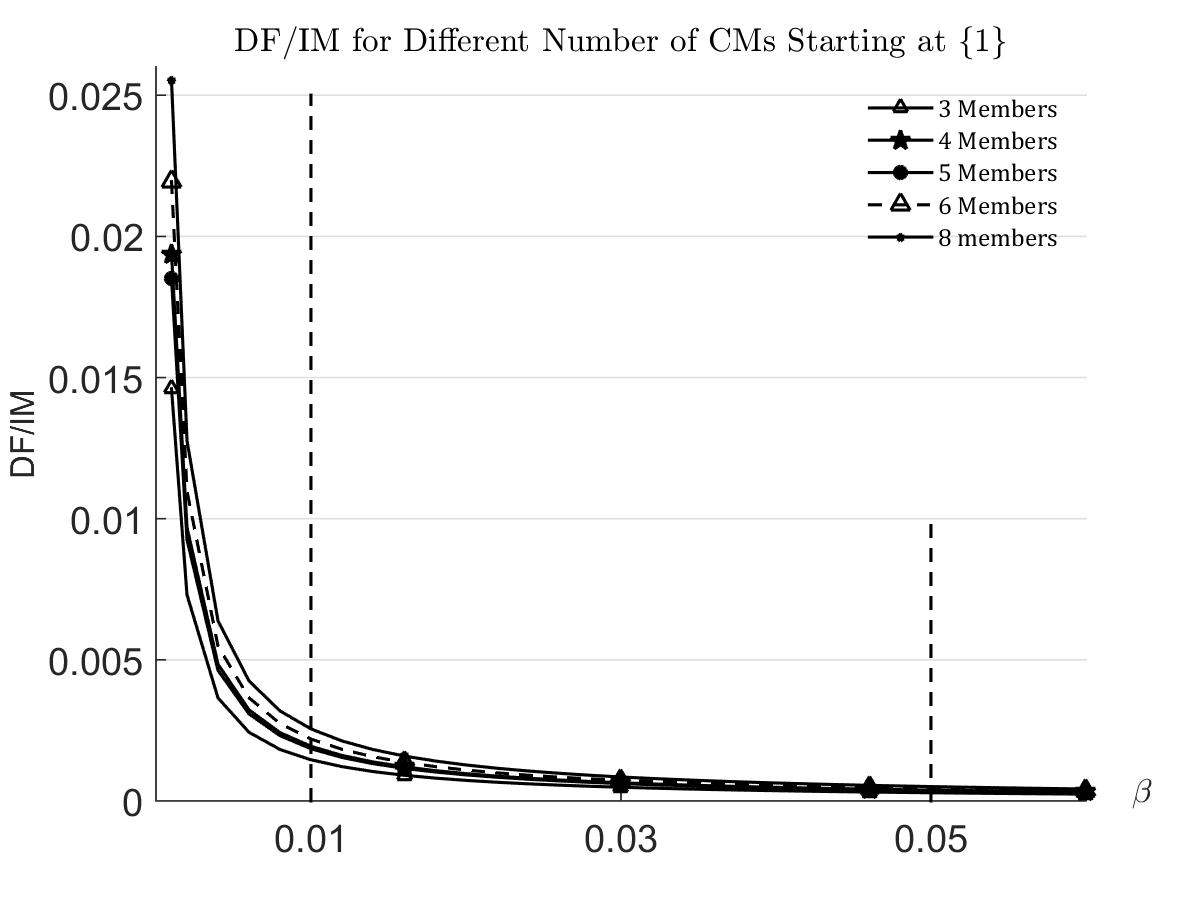}	 %%%%%% WRONG FILE Uncomment above
	\end{subfigure}
	\hfill
	\begin{subfigure}[t]{0.495\textwidth}
		\centering
\includegraphics[width=\linewidth]{DFratioCurveIAS1DiffMem5}	 %%%%%% WRONG FILE Uncomment above
	\end{subfigure}
\caption{Type II DF/IM Ratio}\label{fig:barBTypeII}
  \end{figure}

\begin{figure}[!ht]
	\centering
	\begin{subfigure}[t]{0.495\textwidth}
		\centering
		\includegraphics[width=\linewidth]{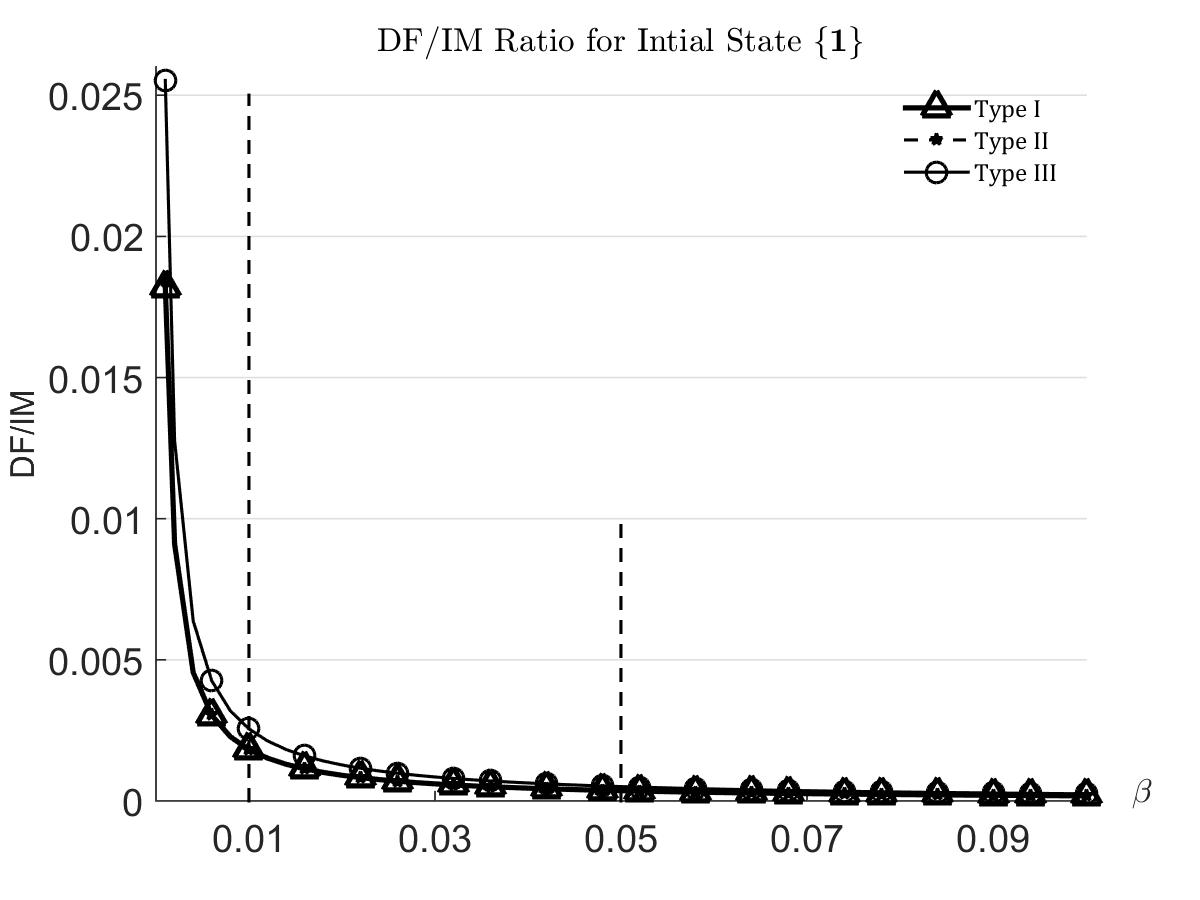}	
	\end{subfigure}
	\hfill
	\begin{subfigure}[t]{0.495\textwidth}
		\centering
		\includegraphics[width=\linewidth]{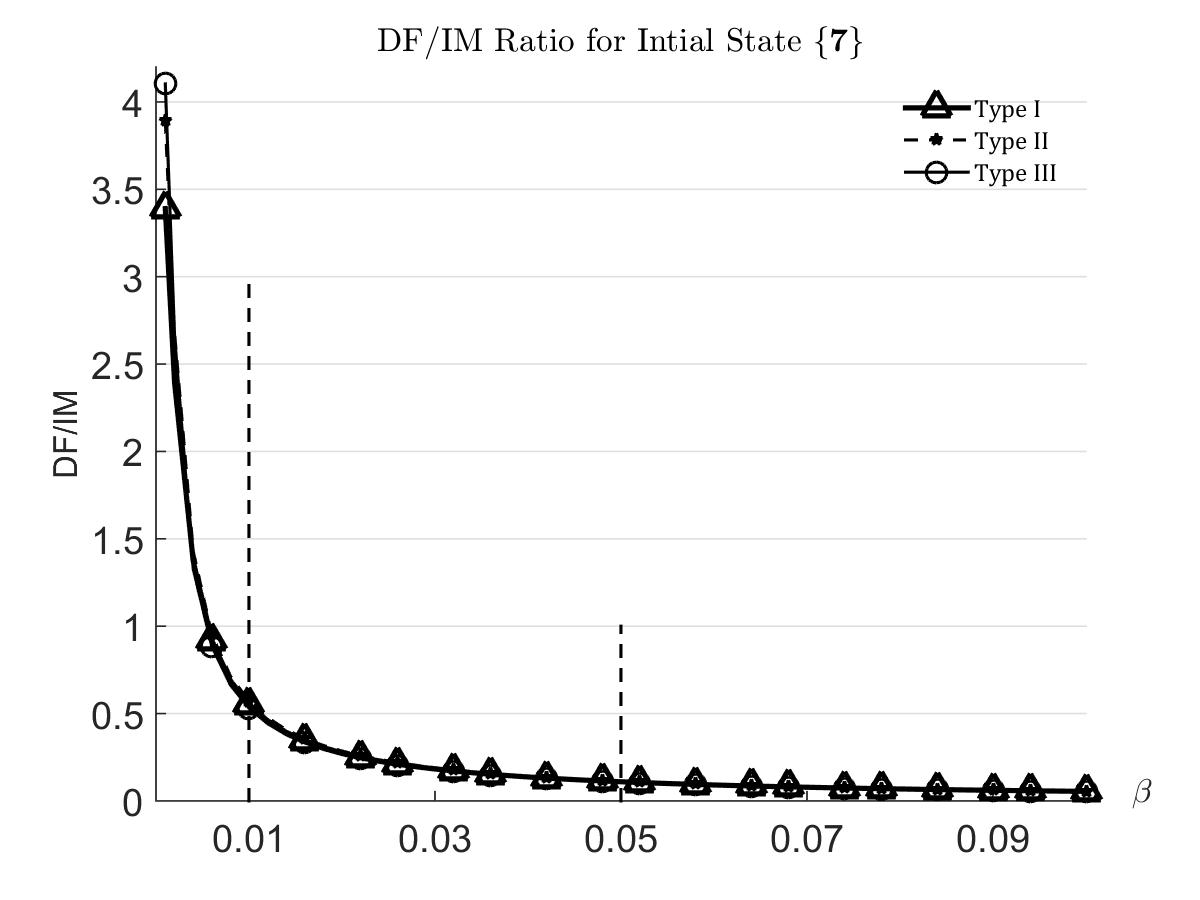}
	\end{subfigure}
\caption{Three Types of DF/IM Ratio for Different $\beta$ with $\alpha=0.01$}\label{fig:DFratio}
\end{figure}

%
%%%%%%%%%%%%%%%%%%% this is a problem too in arxiv
\begin{figure}[!ht]
\centering
	\begin{subfigure}[t]{0.495\textwidth}
		\centering
		\includegraphics[width=\linewidth]{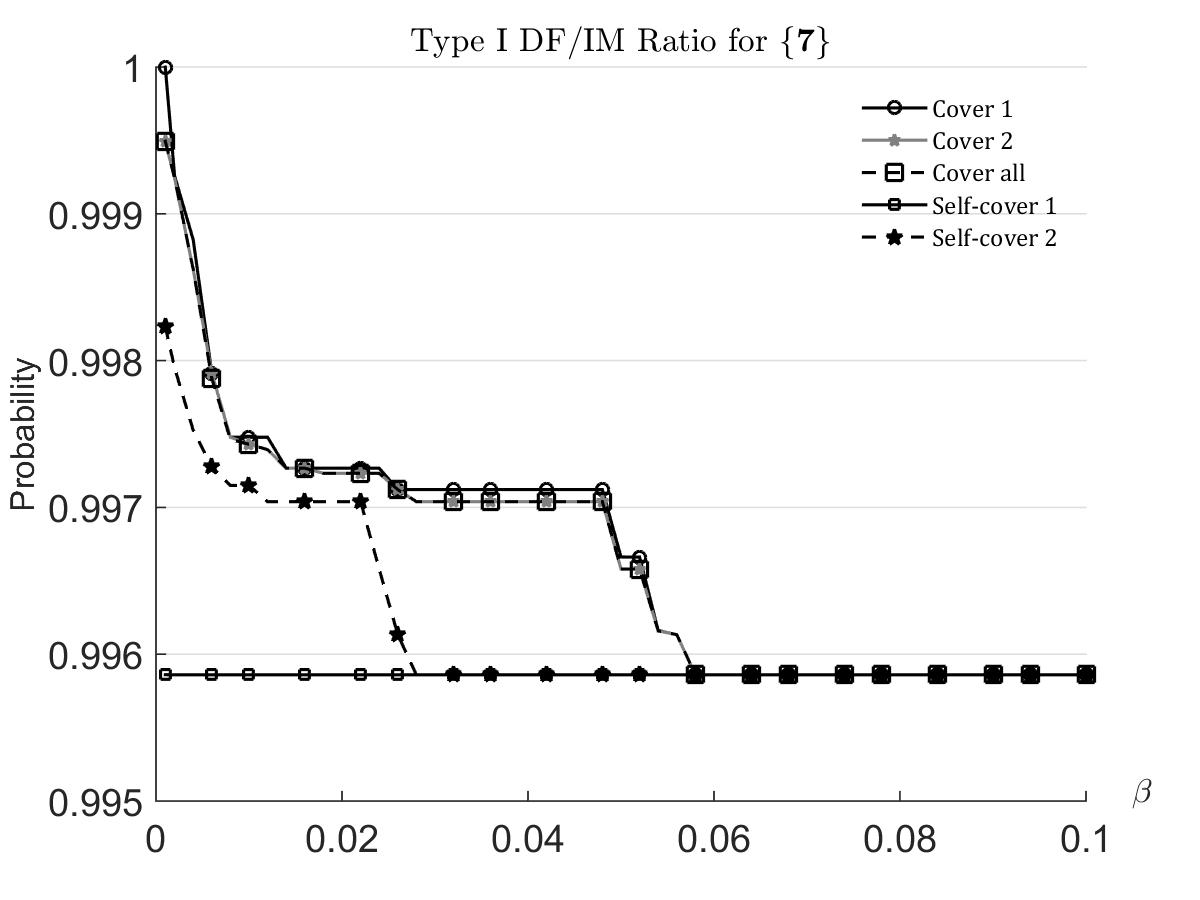}
	\end{subfigure}
	\hfill
	\begin{subfigure}[t]{0.495\textwidth}
		\centering
		\includegraphics[width=\linewidth]{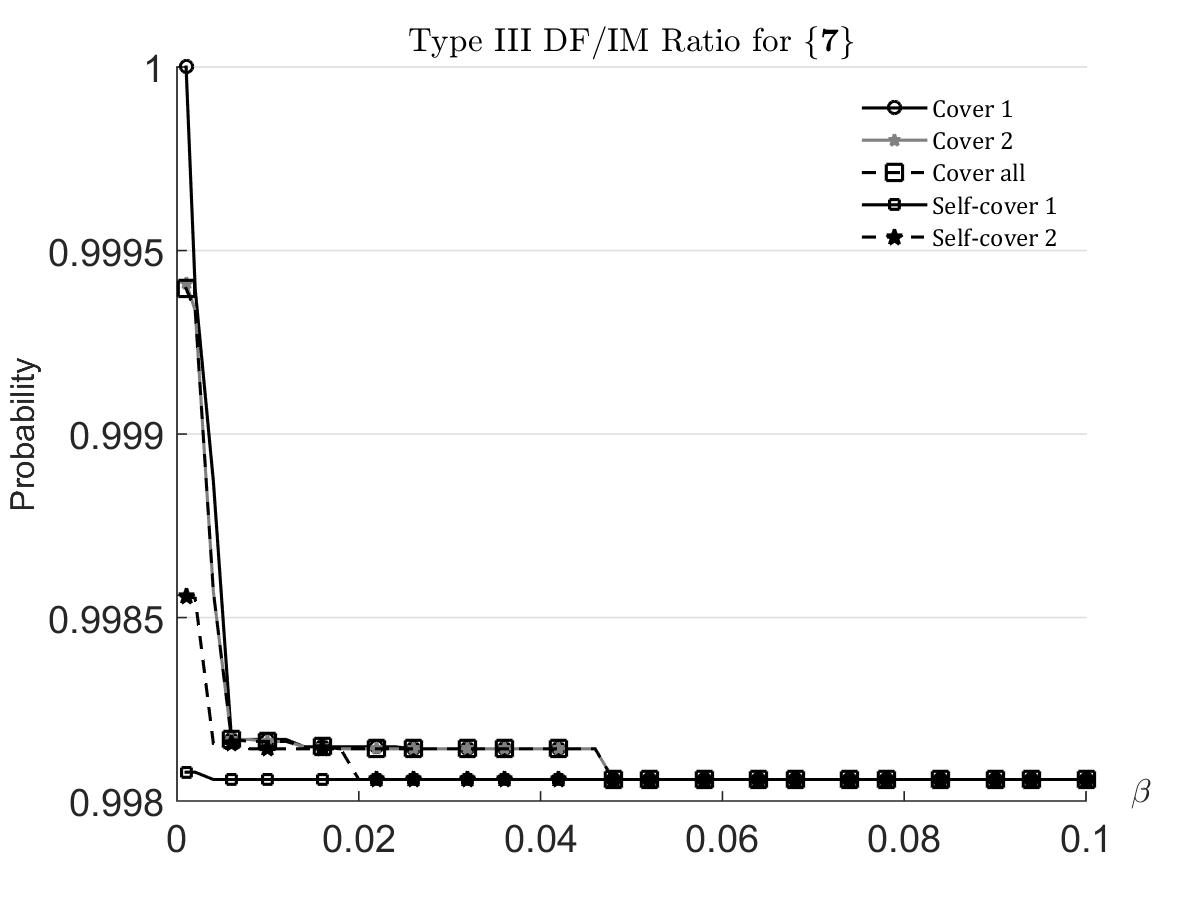}
	\end{subfigure}
	\caption{Type I and III DF Cover Ratio for Initial State $\{\mathbf{7} \}$ }\label{fig:coverRatio}
\end{figure}

\begin{figure}[!ht]
	\centering
	\begin{subfigure}[t]{0.495\textwidth}
		\centering
		\includegraphics[width=\linewidth]{DFratioCurveIAS1DiffMem5}	
	\end{subfigure}
	\hfill
	\begin{subfigure}[t]{0.495\textwidth}
		\centering
		\includegraphics[width=\linewidth]{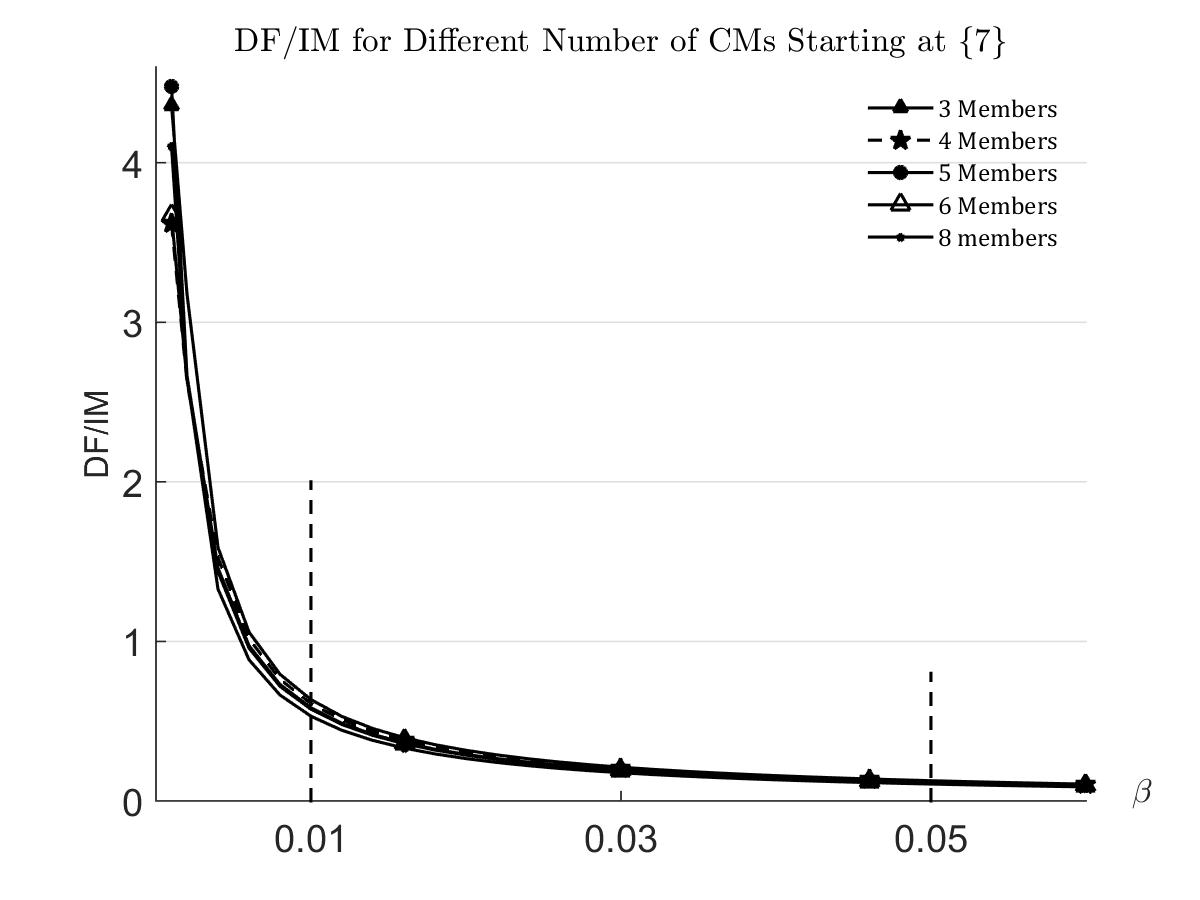}
	\end{subfigure}
	\caption{Type I DF/IM for Different Number of CMs   }\label{fig:DFratioDiffMem}
\end{figure}

%%%%%%%%%%%%%%%%%%%%%%%%%%%%%%%%%%%%%%%%%%%%%%%%%%%%%%%%%%%%%%%%%%%%%%%%%%%%%%%%%%%%%%%%%%%
\begin{appendix}

\section{Conditional Average Value at Risk}\label{sec:avar}
We take  underlying probability space is $(\Omega,\sF,\mathbb{F},P)$, as before.
Recall that the static Average Value at Risk ($\avar$) at the significance level $\alpha \in (0,1]$ is defined as
\begin{equation}\label{eq:defAVAR00}
\avar_\alpha(X):=\frac{1}{\alpha} \int_{0}^{\beta} \var_\beta(X)d\beta,
\end{equation}
for any $X\in L^1(\Omega)$. Following general theory of coherent risk measures, $\avar$ admits the following robust representation
\begin{equation}\label{eq:avarRobust}
  \avar_\alpha(X) = \sup_{Q\in\cQ} E^Q[-X],
\end{equation}
where $\cQ$ is the set of probability measures $Q$ absolutely continuous with respect to $P$, and such that $dQ/dP\leq1/\alpha$.

Let us consider a $\sigma$-algebra $\sG\subset\sF$. In what follows we will use the following notations
\begin{align*}
 \cQ_\sG & = \Set{ Q \prec P \Mid Q=P \textrm{ on } \sG}, \\
  \cQ_\sG^\alpha & = \Set{ Q\in \cQ_\sG \Mid \frac{dQ}{dP}\leq\frac{1}{\alpha}}.
\end{align*}
While the theory of conditional (coherent) dynamic risk measures is a well-established field (c.f.~\cite{DetlefsenScandolo2005,BCDK2013}), we will derive here some specific technical results that are used in this paper.
By analogy to the static case \eqref{eq:avarRobust} (see also \cite[Example~1]{DetlefsenScandolo2005}),  we define the conditional $\avar$ as follows
\begin{equation}\label{eq:condAVARRobust}
  \avar_\alpha(X\mid\sG) := \esssup_{Q\in\cQ_\sG^\alpha} E^Q[-X\mid\sG].
\end{equation}
We will show, in Theorem~\ref{th:maximazer} below,  that there exists a maximizer in \eqref{eq:condAVARRobust}, and we will derive its explicit form. For a similar result in the static case see \cite[Remark~4.48]{FollmerSchiedBook2004}; see also \cite{Cherny2009a} for a similar approach in a dynamic setup.

We start with some definitions and auxiliary results.

Similar to the unconditional case, a \textit{quantile} is an inverse of the conditional cumulative distribution function $F_X(s):=P(X\leq s\mid \sG)$, with $s\in L^0(\sG)$. For a general theory of conditional inverse functions we refer to \cite[Section~A.2]{BCDK2013}.
Note that  $F_X(\,\cdot\,):L^\infty(\sG)\to L^\infty(\sG)$ is increasing, and right-continuous, an hence $G:L^\infty(\sG)\to L^\infty(\sG)$ is an inverse of $F$ if
\begin{equation}\label{eq:defInverse}
F_X^-(G(s))\leq s\leq F(G(s)), \quad s\in L^\infty(\sG),
\end{equation}
where $F^-_X(s):=P(X <  s\mid \sG)$ is the left-continuous version of $F$.

Note that for any $s_1,s_2\in L^\infty(\sG)$, such that $s_1<s_2$, we have
\begin{equation}\label{eq:apend-s1s2}
  F(s_1)\leq F^-(s_2).
\end{equation}

Generally speaking, the function $F$ may have infinitely many inverses, in particular, this may be the case if the random variable $X$ has discrete components. On the other hand, if the $X$ is a continuous random variable, then the inverse $G$ is unique. Among all inverse functions of $F$ we will focus mostly on two of them the left-inverse, and the right-inverse given respectively by
\begin{align}
  F^{-1,-}_X(u) & := \essinf\set{s\in L^0(\sG) \mid F_X(s) \geq u} \label{eq:apenLInvInf}\\
                & = \esssup\set{s\in L^0(\sG) \mid F_X(s) < u},  \label{eq:apenLInvSup} \\
  F^{-1,+}_X(u) & := \essinf\set{s\in L^0(\sG) \mid F_X(s) > u} \label{eq:apenRInvInf}\\
                & = \esssup\set{s\in L^0(\sG) \mid F_X(s) \leq u}. \label{eq:apenRInvSup}
\end{align}
It can be shown (see \cite[Section~A.2]{BCDK2013}) that, $F^{-1,\pm}_X$ are indeed inverses of $F$. Moreover, for any inverse function $G$ of $F$
\begin{equation}\label{eq:inverse2}
  F^{-1,-}\leq G\leq F^{-1,+}.
\end{equation}
and $F^{-1,-}$ is the left-continuous version $G$, and respectively $F^{-1,+}$ is the right-continuous version of $G$.

As mentioned above, a conditional $\alpha$-quantile of a random variable $X$, is defined as the value $G(\alpha)$, with $G$ being an inverse function of the conditional cumulative function $F_X$. The functions $F^{-1,-}(\alpha)$ and $F^{-1,+}(\alpha)$ are called the lower, and respectively, the upper $\alpha$-quantile of $X$. To simplify the notations, we will denote the conditional quantiles by $q_\alpha(X|\sG)$, the lower quantile by $q_\alpha^-(X|\sG)$, and the upper quantile by $q_\alpha^+(X|\sG)$. Sometimes, we will simply write $q^\pm_\alpha(X)$, if no confusions arise. To indicate that the quantile is taken with respect to a measure $P$ we will add the superscript $P$, and  write $q_\alpha^{P}$.

\begin{lemma}\label{lemma:appendQuantile1}
The following representations hold true for the upper and lower conditional quantiles,
\begin{align}
  q^-_\alpha(X\mid\sG) & = \esssup\set{s\in L^0(\sG) \mid P(X<s|\sG) < \alpha},  \label{eq:apenLInvSup2} \\
 q^+_\alpha(X\mid\sG) & =  \esssup\set{s\in L^0(\sG) \mid P(X<s|\sG) \leq \alpha }. \label{eq:apenRInvSup2}
\end{align}
\end{lemma}
\begin{proof}
We will show only \eqref{eq:apenRInvSup2}, and \eqref{eq:apenLInvSup2} is proved similarly.
Let $a$ and $b$ denote the right hand side of \eqref{eq:apenRInvSup2} and \eqref{eq:apenRInvSup} respectively.
Then, since $P(X<s |\sG) \leq P(X\leq x | \sG)\leq \alpha$, we have that
$$
\Set{s\in L^0(\sG) \mid P(X<s|\sG) < \alpha} \supset \Set{s\in L^0(\sG) \mid P(X<s|\sG) \leq \alpha},
$$
and thus $a\geq b$. Assume that $a>b$ on a set $A_0\in\sG$ such that $P(A_0)>0$.
Then, by the definition of $\esssup$, there exists $s_1\in L^\infty(\sG)$, such that $ b < s_1 $ on $A_0$, and $P(X<s_1|\sG)\leq \alpha$.
Consequently, there exists  $s_0\in L^\infty(\sG)$ such that $b<s_0<s_1$ on $A_0$.
In view of \eqref{eq:apend-s1s2},  we have that
$$
\1_{A_0}P(X\leq s_0 | \sG)\leq 1_{A_0} P(X<s_1|\sG) \leq 1_{A_0}\alpha.
$$
Hence, $1_{A_0}b\geq 1_{A_0} s_0$, which leads to contradiction.
\end{proof}

\begin{remark}\label{remark:append1}
It is useful to note that, by \eqref{eq:defInverse},
$$
0\leq P(q^-_\alpha(X|\sG) < X < q^+_\alpha(X|\sG) \mid \sG) \leq \alpha - \alpha = 0,
$$
and thus
\begin{equation}\label{eq:appendQuant9}
P\left(q^-_\alpha(X|\sG) < X < q^+_\alpha(X|\sG) \mid \sG \right) = 0,
\end{equation}
for  $X\in L^\infty(\sF)$.
\end{remark}

\begin{lemma}\label{lemma:appendquant2}
For any $a,m\in L^\infty(\sG)$, such that $a<0$ such that $X/a\in L^\infty(\sF)$, we have
\begin{align}
a q^\mp_{1-\alpha}(X/a\mid\sG) & = q^\pm_\alpha(X\mid\sG), \label{eq:quant4}\\
q^\pm_\alpha(X-m\mid\sG)  & = q^\pm_\alpha(X\mid\sG)-m.  \label{eq:quant5}
\end{align}
\end{lemma}
\begin{proof}
  By \eqref{eq:apenLInvSup2}, we have that
\begin{align}\label{eq:appendquant6}
 a q^{-}_{1-\alpha}(X/a\mid\sG) & = \essinf \Set{as \mid P(X\geq as \mid\sG) > \alpha},
\end{align}
from which, using \eqref{eq:apenRInvInf}, the identity \eqref{eq:quant4}, corresponding to lower quantile on the left hand side, follows at once.
Now, the identity \eqref{eq:apenRInvInf} with upper quantile on left hand sides follows clearly.

Identity \eqref{eq:quant5}  is implied  directly from the definition of the conditional upper quantile.

This concludes the proof.

\end{proof}

\begin{lemma}\label{lemma:E1}
Assume that $Q\in \cQ_\sG$, then
\begin{enumerate}[i)]
  \item  $E[dQ/dP\mid \sG]=1$.

\item   for any $Y\in L^{1}(\sF)$
  \begin{equation}\label{eq:BFormula}
    E^Q[Y\mid \sG ] = E\left[Y \frac{dQ}{dP} \Mid \sG\right],
  \end{equation}
and
\begin{equation}\label{eq:bFormulaInverse}
E[Y \mid \sG] = E[ \1_N Y \mid\sG] + E^Q[\varphi Y \mid \sG],
\end{equation}
where $N = \set{dQ/dP=0}$, $\varphi=dP/dQ$, with the conventions $\varphi=\infty$ if  $dQ/dP=0$, and $0\cdot\infty=0$.
\end{enumerate}
\end{lemma}
\begin{proof} For simplicity, let us assume that $\sG$ is generated by a countable partition $\cD$. \\
%The general proof is carried by usual technics.
i) Let $D\in\cD$. Then, $P(D)=Q(D)$, and thus
\begin{align*}
  E[dQ/dP | D]  & =  \frac{E[dQ/dP \cdot \1_D]}{P(D)} =  \frac{E[dQ/dP \cdot \1_D]}{Q(D)} \\
  & = \frac{E^Q[dQ/dP \cdot dP/dQ \cdot \1_D]}{Q(D)} = \frac{E^Q[\1_D]}{Q(D)} = 1.
\end{align*}

\noindent
ii)
The equality \eqref{eq:BFormula} follows immediately by i) and the abstract Bayes theorem
\[
E^Q[Y\mid \sG ] = \frac{1}{E\left[\frac{dQ}{dP} \mid \sG\right]}E\left[Y \frac{dQ}{dP} \Mid \sG\right].
\]

For $D\in\cD$, we have
\begin{align*}
  E[Y \mid D] & = \frac{E[Y\1_D]}{P(D)} = \frac{E[Y\1_D \1_N]}{P(D)} + \frac{E[Y\1_{D\cap N^c} ]}{P(D)} \\
   & = E[Y \1_N \mid D] + \frac{\int_{{D\cap N^c}} Y \varphi dQ }{Q(D)} \\
   & = E[Y \1_N \mid D] + \frac{\int_{{D}} Y \varphi dQ }{Q(D)},
   \end{align*}
  where in the last equality we used that $Q(N)=0$. Thus, \eqref{eq:bFormulaInverse} is proved.
\end{proof}

In what follows we fix a $Q\in\cQ_\sG$, and put $N=\set{dQ/dP=0}$, $\varphi:=dP/dQ$, with the convention $\varphi=\infty$ on $N$.

Let us denote by $\cA$ the set of all $\sF$-measurable random variables with values in $[0,1]$.
We denote by $\cA^0$ the set of all random variables $\psi\in\cA$ such that there exists $c\in L^0(\sG)$, $c\geq0$ and
  \begin{equation}\label{eq:Psi0}
    \psi=
    \begin{cases}
      1, & \textrm{on \ } \set{\varphi>c} \\
      0, & \textrm{on \ } \set{\varphi<c}.
    \end{cases}
  \end{equation}

\begin{lemma}\label{lemma:ZstarinA0}
Let $\alpha_0\in L^0(\cG)$, such that $\alpha_0\in(0,1)$, and define
\begin{align*}
c^{\sharp} & := q^Q_{1-\alpha_0}(\varphi \mid\sG), \\
\psi^\sharp & := \1_{\varphi>c^\sharp} + \kappa \1_{\varphi=c^\sharp},
\end{align*}
where
$$
\kappa: =
\begin{cases}
 0 , & \text{if } Q(\varphi=c^\sharp | \sG)=0, \\
 \frac{\alpha_0-Q(\varphi>c^\sharp  |  \sG)}{Q(\varphi=c^\sharp | \sG)}, & \mbox{otherwise}.
\end{cases}
$$
Then,   $\psi^\sharp\in\cA^0$  and $E^Q[\psi^\sharp \mid \sG] = \alpha_0$.
\end{lemma}
\begin{proof}
The equality $E^Q[\psi^\sharp \mid \sG] = \alpha_0$ follows directly from the definition of $\psi^\sharp$. It remains to show that $\kappa\in[0,1]$.
Indeed, since $c^\sharp$ is a conditional quantile, by \eqref{eq:defInverse} we deduce
\begin{align*}
 Q(\varphi=c^\sharp | \sG) & = F_\varphi(c^\sharp) - F_\varphi^-(c^\sharp) \\
  & \geq  F_\varphi(c^\sharp) + \alpha_0 - 1\\
  &  = \alpha_0 - Q(\varphi>c^\sharp \mid \sG),
\end{align*}
and thus $\kappa \leq 1$. On the other hand, by \eqref{eq:defInverse}, $Q(\varphi>c^\sharp \mid \sG) = 1-F_\varphi(c^\sharp)\leq 1-(1-\alpha_0) = \alpha_0$, and hence $\kappa\geq 0$.
This completes the proof.
\end{proof}

\begin{lemma}  \label{lemma:NP-Optimizer} \mbox{}
\begin{enumerate}[a)]
  \item For any $\psi^0\in\cA^0$, and $\psi\in\cA$,
\begin{equation}\label{eq:Psi}
  E^Q[\psi\mid \sG]   \leq E^Q[\psi^0\mid \sG] \quad \Rightarrow \quad E^P[\psi\mid\sG] \leq E^P[\psi^0 \mid \sG].
\end{equation}

\item If $\psi^0\in\cA$ satisfies \eqref{eq:Psi} for every $\psi\in\cA$, then $\psi^0\in\cA^0$.
\end{enumerate}
\end{lemma}
\begin{proof}
a) Let $\psi^0\in\cA^0$ and $\psi\in\cA$, then, $\psi^0-\psi\geq0$ on $N$, and $(\psi^0-\psi)(\varphi-c)\geq0$.
Consequently, in view of Lemma~\ref{lemma:E1}.(ii), we have
\begin{align*}
  E^P\left[\psi^0-\psi\Mid \sG\right] & =  E[(\psi^0-\psi) 1_N \mid \sG] + E^Q\left[ \varphi (\psi^0-\psi)\Mid \sG\right] \\
  & \geq cE^Q\left[\psi^0-\psi\Mid \sG\right] \geq0,
\end{align*}
which proves implication \eqref{eq:Psi}.

\smallskip
\noindent
b) Assume that $\psi^*\in\cA$ is such that for any $\psi\in\cA$
\begin{equation}\label{eq:psistar}
  E^Q[\psi\mid \sG]   \leq E^Q[\psi^*\mid \sG] \quad \Rightarrow \quad E^P[\psi\mid\sG] \leq E^P[\psi^* \mid \sG].
\end{equation}
Let $\alpha_0:=E^Q[\psi^*\mid\sG]$. If $\alpha_0=0$ or $\alpha_0=1$, the statement is trivial.

For $0<\alpha_0<1$,  we take $\psi^\sharp\in\cA^0$ as in Lemma~\ref{lemma:ZstarinA0}, for which we have
$$
\alpha_0=E^Q[\psi^*\mid\sG] = E^Q[\psi^\sharp \mid \sG].
$$
By \eqref{eq:Psi}, we get that $E^P[\psi^*\mid\sG] \leq E^P[\psi^\sharp \mid \sG]$, while by \eqref{eq:psistar} we have the opposite inequality, and thus
$$
E^P[\psi^*\mid\sG] = E^P[\psi^\sharp \mid \sG].
$$
Let $g:=\psi^\sharp-\psi^*$. Then, $g\geq 0$ on $N= \set{\varphi=\infty}$, and $g\cdot(\varphi-c^\sharp)\geq0$ $Q$-a.s. Consequently, using \eqref{eq:bFormulaInverse}, we conclude
\begin{align*}
0 & = E^P[g\mid\sG] - c^\sharp E^Q[g\mid\sG] = E^P[g\cdot\1_N\mid\sG] + E^Q[g\cdot(\varphi-c^\sharp)\mid\sG] \\
 & \geq E^Q[g\cdot(\varphi-c^\sharp)\mid\sG]\geq0.
\end{align*}
Thus, $\psi^*=\psi^\sharp$ is $Q$-a.s.  on $\varphi\neq c^\sharp$, and hence $P$-a.s. on $\varphi\neq c^\sharp$.

The general case of $\alpha_0\in[0,1]$, follows from the above cases. We skip the details. An interested reader may contact the authors.
 \end{proof}

Finally, we present the result about the maximizer in the robust representation \eqref{eq:condAVARRobust} of the conditional $\avar$.
\begin{theorem}\label{th:maximazer}
For any $X\in L^\infty(\sF)$,
\begin{equation}\label{eq:condAVARRobust2}
  \avar_\alpha(X\mid\sG) = \esssup \Set{ E[-XZ\mid\sG]\Mid Z\in\sF,\ 0\leq Z\leq 1/\alpha,\ E(Z\mid\sG)=1}.
\end{equation}
Then, a maximizer $Z^*$ in the right hand side of \eqref{eq:condAVARRobust2} exists, and it is given by
\begin{equation}\label{eq:Zstar}
Z^* = \frac{1}{\alpha} \left(\1_{X<q^\pm_\alpha(X\mid \sG)} + \varepsilon \1_{X=q^\pm_\alpha(X\mid \sG)} \right),
\end{equation}
where
$$
\varepsilon =
\begin{cases}
  0, & \mbox{if } P(X=q^\pm_\alpha(X\mid \sG))=0 \\
  \frac{\alpha-P(X<q^\pm_\alpha(X\mid \sG))}{P(X=q^\pm_\alpha(X\mid \sG))}, & \mbox{otherwise}.
\end{cases}
$$
Consequently, a maximizer $Q^*$ in the right hand side of \eqref{eq:condAVARRobust} exists, and it is given by $dQ^*/dP=Z^*.$
\end{theorem}

\begin{proof}%[Proof of Theorem~\ref{th:maximazer}]
If $\alpha=1$, then $\cQ_\sG=\set{P}$ and the statement is obvious. In particular, $\{Z\in\sF,\ 0\leq Z\leq 1/\alpha,\ E(Z|\sG)=1\}=\{Z\in\sF, \ Z=1, P-a.s.\}$

Next, assume that $\alpha\in(0,1)$. First, we observe that \eqref{eq:condAVARRobust2} follows from \eqref{eq:condAVARRobust} and \eqref{eq:BFormula}.

In order to prove \eqref{eq:Zstar} we first assume that $X<0$. We consider the measure $\widetilde{P}\sim P$, such that
$d\widetilde{P}/dP = X/E[X|\sG]$. Note that since $E[X/E[X|\sG] \mid \sG]=1$,  and $\widetilde{P}\in\cQ_\sG$.

For any $Q\in\cQ_\sG$, in view of \eqref{eq:BFormula} from Lemma~\ref{lemma:E1}(ii), we  have, upon letting $Z=dQ/dP$,
\begin{align*}
  E^Q[X\mid\sG]  & =  E[XZ\mid \sG] = E\left[\frac{X}{E[X\mid\sG]} \ E[X\mid\sG] Z \mid \sG \right] \\
  & =  E[X\mid\sG] \cdot E\left[\frac{X}{E[X\mid\sG]} Z\mid \sG\right] = E[X\mid\sG] \cdot E^{\widetilde{P}} \left[Z\mid\sG\right].
\end{align*}
Consequently, in view of \eqref{eq:condAVARRobust2}, we see that
\begin{align}
\avar_\alpha(X\mid\sG) & =\frac{E(-X\mid \sG)}{\alpha}\esssup \Set{ E^{\widetilde P}[Z\mid\sG]\Mid Z\in\sF,\ 0\leq Z\leq \frac{1}{\alpha},\ E(Z|\sG)=1}  \nonumber \\
& = \frac{E(-X\mid \sG)}{\alpha}\esssup \Set{ E^{\widetilde P}[\psi\mid\sG]\Mid \psi\in\sF,\ 0\leq \psi\leq 1,\ E(\psi|\sG) = \alpha} \nonumber \\
& = \frac{E(-X\mid \sG)}{\alpha}\esssup \Set{ E^{\widetilde P}[\psi\mid\sG]\Mid \psi\in\sF,\ 0\leq \psi\leq 1,\ E(\psi|\sG)\leq\alpha} . \label{eq:ap5}
\end{align}
Next, we apply Lemma~\ref{lemma:NP-Optimizer}, and for clarity, we denote the probability measures $Q,P$ from this lemma by $\check{Q}$ and $\check{P}$.  respectively. By taking $\check Q=P$, $\check P=\widetilde{P}$, and $\psi^0=\psi^\sharp$ with $\psi^\sharp$ given in Lemma~\ref{lemma:ZstarinA0}, we have that
$$
E^{\check{P}}[\psi \mid \sG] \leq E^{\check{P}}[\psi^\sharp \mid \sG],
$$
for any $\psi\in L^0(\sG),$  such that $0\leq \psi\leq 1$, and $E^{\check{Q}}[\psi|\sG]\leq E^{\check{Q}}[\psi^\sharp| \sG] = \alpha$.
Thus, the $\esssup$ in \eqref{eq:ap5} is attained at $\psi^\sharp$, and consequently we get
\begin{align*}
\avar_\alpha(X\mid\sG) & = \frac{E(-X\mid \sG)}{\alpha} E^{\widetilde{P}}[\psi^\sharp \mid \sG]\\
        & = \frac{1}{\alpha} E[-X\psi^\sharp \mid \sG].
\end{align*}
By Lemma~\ref{lemma:appendquant2}, with  $\varphi=\check{P}/\check{Q} = d\widetilde{P}/dP$, and $c^\sharp=q_{1-\alpha}^{P,\pm}(\varphi\mid\sG)$
$$
E^P[X \mid\sG] c^\sharp= q_{\alpha}^\mp(X|\sG).
$$
Consequently, $\psi^\sharp/\alpha=Z^*$.

Finally, we will show that theorem holds true for an arbitrary $X\in L^\infty(\sF)$. For a fixed $X\in L^\infty(\sF)$, we consider $m\in\bR$ such that  $X'=X-m<0$.
Thus,  by applying to $X'$ the above result, using \eqref{eq:quant5} and cash-additivity of $\avar$, and we deduce
\begin{align*}
  \avar_\alpha(X\mid\sG) & = \avar_\alpha(X-m\mid \sG) - m  \\
 & = \frac{1}{\alpha}E[-(X-m)\big(\1_{X-m<q^\pm_\alpha(X-m|\sG)} + \widehat{\kappa} \1_{X-m<q^\pm_\alpha(X-m|\sG)}\big) \mid\sG] - m,
 \end{align*}
 where
\begin{align*}
\widetilde{\kappa}= \frac{\alpha-P(X-m<q^\pm_\alpha(X-m|\sG))}{P(X-m=q^\pm_\alpha(X-m)|\sG)} =
\frac{\alpha-P(X<q^\pm_\alpha(X|\sG))}{P(X=q^\pm_\alpha(X)|\sG)}= \varepsilon,
\end{align*}
on the set $P(X=q^\pm_\alpha(X)|\sG)\neq0$.
Consequently,
\begin{align*}
  \avar_\alpha(X\mid\sG)  & =  E[-X Z^*  \mid\sG] +\frac{m}{\alpha} E[\1_{X<q^\pm_\alpha(X|\sG)} + \varepsilon\1_{X<q^\pm_\alpha(X|\sG)}\mid\sG] - m \\
  & =  E[-X Z^*  \mid\sG] + mE[Z^*\mid\sG]-m = E[-X Z^*  \mid\sG].
\end{align*}

\end{proof}

\begin{remark}\label{remark:appendAfterTh}
i) Generally speaking, one can prove that the Theorem~\ref{th:maximazer} holds true with $q^\pm_\alpha$ being replaced by any conditional $\alpha$-quantile of $X$, in which case, in view of \eqref{eq:appendQuant9} $\varepsilon=0$ if $q_\alpha\neq q^\pm_\alpha$.

\noindent
ii) Clearly, if $X$ is a continuous random variable, then the maximizer $Z^*$ in Theorem~\ref{th:maximazer} is unique. Generally speaking, the maximizer $Z^*$ is not unique. Nevertheless, the value of $\avar_\alpha(X|\sG)$ is unique, and does not depend on the choice of the maximizer $Z^*$.

\noindent
iii) There are several other representations of conditional $\avar$. In particular, one can show that a representation similar to \eqref{eq:defAVAR00} holds true for conditional case too. This is out of scope of this paper, and we will skip such derivations here.
\end{remark}

\section{Discrete Time Markovian Structure Model of Credit Migrations}\label{sec:MC}

\subsection{Theory}

Let $(\Omega, \sF, P)$ be the underlying (statistical) probability space.
We denote by $\gls{Ri}_t,\ t\in \mathcal{T}$  the credit ratings process of \gls{CMi}, $i\in \cI.$ We assume that process $R^i$ is a time homogeneous Markov chain taking values in the finite state space, say \gls{cRi}, representing possible credit ratings of the $i$th member. Without loss of generality we take $\cR^i=\set{1,\ldots,K^i},$ where $1$ corresponds to the highest (the best) credit rating, and where $K^i$ corresponds to the default state. We assume that the default state $K^i$ is absorbing.  The transition matrix of process $R^i$ is denoted by $\mP^i=[p^i_{x_i,y_i}]_{x_i,y_i\in \cR^i}$. Thus,  $p^i_{x_i,y_i}= P(R^i_{t_1}=y_i \mid R^i_0=x_i)$, for $x_i\neq y_i$.
Typically, the matrix $\mP^i$ can be obtained from the data provided by the rating agencies; see Section \ref{sec:Piestimate} for details.

Now, let us consider a system of $I$ algebraic equations in unknowns $p ^{x}_{y}$, where
$x:= (x_1,x_2,\ldots ,x_I),\ y:= (y_1,y_2,\ldots ,y_I) \in \cR:=  \cR^1\times \cdots \times \cR^I$
and $x\ne y$:

\begin{align}\label{system-1}
p^i_{x_i,y_i}&=\sum_{y_j\in \cR^j,j\ne i}p^{x_1,\ldots,x_i,\ldots,x_I}_{y_1,\ldots,y_i,\ldots,y_I}, \quad \forall x_j \in \cR^j,\, j\ne
i,\quad \forall x_i,\ y_i \in \cR^i,\, x_i\ne y_i, \nonumber \\ i&=1,2,\ldots,I.
\end{align}

Similarly as in \cite{BJN-3}, where the continuous time case is studied,  one can show that the above system admits at least one solution $p^{x}_{y}$ such that, if we define
\begin{equation}\label{diag1}
p ^x_x=1-\sum _{y\in \cR, y\ne x}p ^{x}_{y},
\end{equation}
then the matrix
$$
\mP:= [p^{x}_{y}]_{x,y\in \cR}
$$
is a transition matrix of a Markov chain, say $R=(X^1,\ldots,X^I)$, with state space $\cR,$ and that each component $X^i$ is also Markov chain with transition matrix \gls{mPi}.

The process $R$ describes the joint evolution of credit ratings of all CMs of the CCP, and has the property that its components are Markovian and with the same transition laws as the individual credit migration processes $R^1,\dots,R^I$. If the distribution of $(X^1_0,\ldots,X^I_0)$ is the same as distribution of $(R^1_0,\ldots,R^I_0)$ then, in the terminology of \cite{BJN-2016}, the process $R$ is known as the strong Markovian structure for the processes $R^1,\ldots,R^I$.

In summary, the process $R$ is a time homogenous Markovian model for joint evolution of credit ratings of all CMs of the CCP, subject to the marginal constraints that laws of the components of $R$ match the laws of the individual credit migrations of the CMs.

\begin{remark}
If we insist on continuous time simulation then the equations above change to the  of $I$ algebraic equations in unknowns $\lambda ^{x}_{y}$, where
$x:= (x_1,x_2,\ldots ,x_I),\ y:= (y_1,y_2,\ldots ,y_I) \in \cR:=  \cR^1\times \cdots \times \cR^I$
and $x\ne y$:

\begin{align}\label{system-2}
\lambda^i_{x_i,y_i}&=\sum_{y_j\in \cR^j,j\ne i}\lambda^{x_1,\ldots,x_i,\ldots,x_I}_{y_1,\ldots,y_i,\ldots,y_I}, \quad \forall x_j \in \cR^j,\, j\ne
i,\quad \forall x_i,\ y_i \in \cR^i,\, x_i\ne y_i, \nonumber \\ i&=1,2,\ldots,I,
\end{align}
where $\lambda^i$s represent marginal generators.

As in \cite{BJN-3}, where the continuous time case is studied,  one can show that the above system admits at least one solution $\lambda^{x}_{y}$ such that, if we define
\begin{equation}\label{diag2}
\lambda ^x_x=-\sum _{y\in \cR, y\ne x}\lambda ^{x}_{y},
\end{equation}
then the matrix
$$
\Lambda:= [\lambda^{x}_{y}]_{x,y\in \cR}
$$
is the generator matrix of a Markov chain, say $R=(X^1,\ldots,X^I)$, with state space $\cR,$ and that each component $X^i$ is also Markov chain with generator matrix $\lambda^i$.

\end{remark}

\begin{remark}
These systems may be changed to time inhomogeneous systems.
\end{remark}

\subsubsection{How to solve the above systems}
Brute force solution of the above systems is impossible. With 8 members and with 7 rating categories for each member, the state space of the Markov chain $R$ contains $7^{8}$ elements. So the corresponding transition matrix (or the generator matrix in the continuous time) is a $7^{8}\times 7^{8}$ matrix.

So, the proposed way to proceed is as follows: In case of the discrete time
\begin{enumerate}
\item For $n=0$, solve the system
\begin{align}\label{system-1-n=0}
p^i_{X_i(0),y_i}&=\sum_{y_j\in \cR^j,j\ne i}p^{X_1(0),\ldots,X_i(0),\ldots,X_I(0)}_{y_1,\ldots,y_i,\ldots,y_I}, \quad \forall y_i\in N(X_i(0)),\ p^{X_1(0),\ldots,X_i(0),\ldots,X_I(0)}_{y_1,\ldots,y_i,\ldots,y_I}\in [0,1], \nonumber \\ i&=1,2,\ldots,I,
\end{align}
where $N(X_i(0))$ is the set of (at most two) closest ratings to $X_i(0)$; for $y_i\notin  N(X_i(0))$ set $p^{X_1(0),\ldots,X_i(0),\ldots,X_I(0)}_{y_1,\ldots,y_i,\ldots,y_I}=0.$ Define
\begin{equation}\label{diag1}
p ^{X(0)}_{X(0)}=1-\sum _{y\in \cR, y\ne {X(0)}}p ^{X(0)}_{y},
\end{equation}
and simulate the first transition $X(0)\rightarrow X(1)$ according to $p^{X_1(0),\ldots,X_i(0),\ldots,X_I(0)}_{y_1,\ldots,y_i,\ldots,y_I}$ computed above.
\item For $n=1$, solve the system
\begin{align}\label{system-1-n=0}
p^i_{X_i(1),y_i}&=\sum_{y_j\in \cR^j,j\ne i}p^{X_1(1),\ldots,X_i(1),\ldots,X_I(1)}_{y_1,\ldots,y_i,\ldots,y_I}, \quad \forall y_i\in N(X_i(1)),\ p^{X_1(1),\ldots,X_i(1),\ldots,X_I(1)}_{y_1,\ldots,y_i,\ldots,y_I}\in[0,1] \nonumber \\ i&=1,2,\ldots,I,
\end{align}
where $N(X_i(1))$ is the set of (at most two) closest ratings to $X_i(1)$; for $y_i\notin  N(X_i(1))$ set $p^{X_1(1),\ldots,X_i(1),\ldots,X_I(1)}_{y_1,\ldots,y_i,\ldots,y_I}=0.$ Define
\begin{equation}\label{diag1}
p ^{X(1)}_{X(1)}=1-\sum _{y\in \cR, y\ne X(1)}p ^{X(1)}_{y},
\end{equation}
and simulate the second transition $X(1)\rightarrow X(2)$ according to $p^{X_1(1),\ldots,X_i(1),\ldots,X_I(1)}_{y_1,\ldots,y_i,\ldots,y_I}$ computed above.
\item And so on for $n\geq 2$.
\end{enumerate}

\begin{remark}
In the continuous time set-up proceed analogously, with obvious modifications.
\end{remark}

\subsection{Estimation of \gls{mPi}s}\label{sec:Piestimate}
Denote by \gls{mPiy} the matrix of one year transition probabilities for credit ratings for \gls{CMi}.
Rating agencies typically provide one year transition probabilities for credit ratings for various obligors, and thus we will assume that \gls{mPiy} are known or observed from market data.
Let $m$ be an integer such that $m\delta_f=\textrm{one year}$; e.g. if the fundamental unit of time is one day, we take $m=252$.
The calibration of the transition matrix $\mP^i$ is done by solving for $\mP^i$  the following matrix equation
\begin{equation}\label{calibration}
\left (\mP^i\right )^m= \gls{mPiy},
\end{equation}
subject to the constraint that $\mP^i$ is a stochastic matrix.

\section{CDS modeling}\label{sec:bigbang}
In April of 2009 the market for credit default swap (CDS) contracts went through some fundamental changes with regard to  the contract conventions. The changes were dubbed as ``Big Bang'', and they resulted in standardizing the  CDS market and supporting  the central clearing of CDS contracts.

Before the Big Bang, CDS contracts used to be quoted in terms of a fair spread, which made the present mark to market MtM of the CDS contract\footnote{ That was taken as the difference between the values of the two legs of the contract.} null for both parties in the contract -- the protection seller and protection buyer. After the ``Big Bang ", all CDS spreads (coupons)  have been standardized as either 100 or 500 basis points (bps), resulting in an exchange of upfront payment so to make the value of the contract at initiation equal to zero. Moreover, under the new convention, the CDS contracts can only terminate on March 20, June 20, September 20 and December 20 in a given year.

Currently, according to the CCP industry standards the marking to market means computing the present upfront payment. Thus, the MtM of a CDS contract is equal to the marked to market upfront payment. We adopt this convention here.

We will now briefly describe the way in which we compute the upfront payment. Towards this end we denote by $\lambda$ be the (constant) intensity of the default time $\phi$ of the reference name underlying a CDS contract. In addition we denote the CDS spread (coupon) as $\kappa$, and the constant recovery rate as $R$. For simplicity, we assume that the discount factor $\beta$ is one.
Given all this, the only relevant flow of information regarding the given CDS contract is the is the natural filtration generated by the default indicator process $H_t=\1_{\phi\leq t}$, $t\geq 0$. We denote this filtration as $\mathbb{H}=(\mathcal{H}_t,\ t\geq 0).$

Then  the upfront payment at time $t\geq 0$, say $S_t$, based on the notional value of the contract equal to 1, is computed as follows:
\begin{align}\label{eq:CDS}
	S_t&=E[R\1_{t<\phi\le T}-(T\wedge \phi -t)\kappa|H_t],\\
	&=\1_{\phi>t}\widetilde{S_t},
\end{align}
where the so called pre-default up-front payment $\widetilde{S_t}$ is
\begin{equation}\label{eq:CDSpre}
\widetilde{S}_t=(e^{-\lambda(T-t)}-1)\frac{\kappa-\lambda R}{\lambda}.
\end{equation}
\end{appendix}

%%%%%%%%%%%%%%%%%%%%%%%%%%%%%%%%%%%%%%%%%%%%%%%%%%%%%%
\section*{Acknowledgments}
Part of the research was performed while Igor Cialenco was visiting the Institute for Pure and Applied Mathematics (IPAM), which is supported by the National Science Foundation. The authors are thankful to Samim Ghamami for useful discussions and helpful remarks.

{\small
\bibliographystyle{alpha}
%\bibliography{MathFinanceMaster-10-06-2017}

\begin{thebibliography}{VCDF{\etalchar{+}}15}

\bibitem[AC17]{ArmentiCrepey2015}
Y.~Armenti and S.~Cr\'epey.
\newblock Central clearing valuation adjustment.
\newblock {\em SIAM J. Financial Math.}, 8(1):274--313, 2017.

\bibitem[ACDP15]{ArmentEtAl2015}
Y.~Armenti, S.~Crepey, S.~Drapeau, and A.~Papapantoleon.
\newblock Multivariate shortfall risk allocation and systemic risk.
\newblock {\em Preprint}, 2015.

\bibitem[ADE{\etalchar{+}}02]{ArtznerDelbaenEberHeathKu2002b}
P.~Artzner, F.~Delbaen, J.-M. Eber, D.~Heath, and H.~Ku.
\newblock Coherent multiperiod risk measurement.
\newblock {\em Preprint}, 2002.

\bibitem[AP11]{AcciaioPenner2010}
B.~Acciaio and I.~Penner.
\newblock Dynamic risk measures.
\newblock {\em In G. Di Nunno and B. {\O}ksendal (Eds.), Advanced Mathematical
  Methods for Finance, Springer}, pages 1--34, 2011.

\bibitem[Arn17]{Arnold2017}
M.~Arnold.
\newblock The impact of central clearing on banks’ lending discipline.
\newblock {\em Forthicoming in Journal of Financial Markets}, 2017.

\bibitem[BC14]{BrunnermeierCheridito2014}
M.~Brunnermeier and P.~Cheridito.
\newblock Measuring and allocating systemic risk.
\newblock {\em Preprint}, 2014.

\bibitem[BCDK16]{BCDK2013}
T.~R. Bielecki, I.~Cialenco, S.~Drapeau, and M.~Karliczek.
\newblock Dynamic assessment indices.
\newblock {\em Stochastics: An International Journal of Probability and
  Stochastic Processes}, 88(1):1--44, 2016.

\bibitem[BCP14]{BCP2014a}
T.~R. Bielecki, I.~Cialenco, and M.~Pitera.
\newblock A unified approach to time consistency of dynamic risk measures and
  dynamic performance measures in discrete time.
\newblock {\em Preprint}, 2014.

\bibitem[BCP17]{BCP2017}
T.~R. Bielecki, I.~Cialenco, and M.~Pitera.
\newblock A survey of time consistency of dynamic risk measures and dynamic
  performance measures in discrete time: {LM}-measure perspective.
\newblock {\em Probability, Uncertainty and Quantitative Risk}, 2017.
\newblock published online.

\bibitem[BJN13]{BJN-3}
T.~R. Bielecki, J.~Jakubowski, and M.~Nieweglowski.
\newblock Intricacies of dependence between components of multivariate {M}arkov
  chains:weak {M}arkov consistency and {M}arkov copulae.
\newblock {\em Electron. J. Probab.}, 18:no. 45, 21, 2013.

\bibitem[BJN16]{BJN-2016}
T.~R. Bielecki, J.~Jakubowski, and M.~Nieweglowski.
\newblock Conditional markov chains: Properties, construction and structured
  dependence.
\newblock {\em Forthcoming Stochastic Processes and their Applications}, 2016.

\bibitem[Boa14]{FSB2014}
Financial~Stability Board.
\newblock {OTC} derivatives market reforms.
\newblock {\em 8th Progress Report on Implementation}, 2014.

\bibitem[BV17]{BeliVaaradi2017}
M.~B{\'e}li and K.~V{\'a}radi.
\newblock A possible methodology for determining the initial margin.
\newblock {\em Forthcoming in Financial and Economic Review}, 2017.

\bibitem[CCS17]{CapponiEtAl2017}
A.~Capponi, W.~A. Cheng, and J.~Sethuraman.
\newblock Clearinghouse default waterfalls: Risk-sharing, incentives, and
  systemic risk.
\newblock {\em Preprint}, 2017.

\bibitem[Che06]{ChernyWVAR2006}
A.~Cherny.
\newblock Weighted {V}@{R} and its properties.
\newblock {\em Finance Stoch.}, 10(3):367--393, 2006.

\bibitem[Che07]{cherny-2006p2}
A.~Cherny.
\newblock Pricing with coherent risk.
\newblock {\em Probability Theory and Its Applications}, 52(3):506--540, 2007.

\bibitem[Che09]{Cherny2009a}
A.~Cherny.
\newblock Capital allocation and risk contribution with discrete-time coherent
  risk.
\newblock {\em Math. Finance}, 19(1):13--40, 2009.

\bibitem[Che17]{ChengPhd2017}
W.~A. Cheng.
\newblock {\em Clearinghouse Default Resources: Theory and Empirical Analysis}.
\newblock PhD thesis, Columbia University, 2017.

\bibitem[CLLW17]{CuiEtAl2017}
Z.~Cui, C.~Lee, Y.~Liu, and K.~Wang.
\newblock Failure and rescue in central clearing counterparty design.
\newblock {\em Preprint}, 2017.

\bibitem[Del00]{Delbaen2000}
F.~Delbaen.
\newblock {\em Coherent risk measures}.
\newblock Scuola Normale Superiore, 2000.

\bibitem[Den01]{Denault2001}
M.~Denault.
\newblock Coherent allocation of risk capital.
\newblock {\em Journal of Risk}, 4(1):1--34, 2001.

\bibitem[Den17]{Deng2017}
B.~Deng.
\newblock Counterparty risk, central counterparty clearing and aggregate risk.
\newblock {\em Forthcoming in Annals of Finance}, 2017.

\bibitem[DS05]{DetlefsenScandolo2005}
K.~Detlefsen and G.~Scandolo.
\newblock Conditional and dynamic convex risk measures.
\newblock {\em Finance and Stochastics}, 9(4):539--561, 2005.

\bibitem[ELW16]{EmbrechtsLiuyWangz2016}
P.~Embrechts, H.~Liuy, and R.~Wangz.
\newblock Quantile-based risk sharing.
\newblock {\em Preprint}, 2016.

\bibitem[EMI12]{EMIR2012}
EMIR.
\newblock {OTC} derivatives, central counterparties and trade repositories.
\newblock {\em Regulation (EU) No 648/2012 of the European Parliament and of
  the Council}, 2012.

\bibitem[EUR14]{EUREX2014}
EUREX.
\newblock How central counterparties strengthen the safety integrity of
  financial markets.
\newblock {\em CCP Recovery and Resolution Conference, Chicago.}, 2014.

\bibitem[Fis03]{Fischer2003}
T.~Fischer.
\newblock Risk capital allocation by coherent risk measures based on one-sided
  moments.
\newblock {\em Insurance: Mathematics and Economics}, 32:135--146, 2003.

\bibitem[FS04]{FollmerSchiedBook2004}
H.~F{\"o}llmer and A.~Schied.
\newblock {\em Stochastic finance. An introduction in discrete time}, volume~27
  of {\em de Gruyter Studies in Mathematics}.
\newblock Walter de Gruyter \& Co., Berlin, extended edition, 2004.

\bibitem[GG16]{GhamamiGlasserman2016}
S.~Ghamami and P.~Glasserman.
\newblock Does {OTC} derivatives reform incentivize central clearing?
\newblock {\em Preprint}, 2016.

\bibitem[Gha15]{Ghamami2014}
S.~Ghamami.
\newblock Static models of central counterparty risk.
\newblock {\em International Journal of Financial Engineering}, 02(02):1550011,
  2015.

\bibitem[Gre14]{Gregory2014Book}
J.~Gregory.
\newblock {\em {C}entral {C}ounterparties: {M}andatory {C}entral {C}learing and
  {I}nitial {M}argin {R}equirements for {OTC} {D}erivatives}.
\newblock Wiley, 2014.

\bibitem[Kal05]{Kalkbrener2005}
M.~Kalkbrener.
\newblock An axiomatic approach to capital allocation.
\newblock {\em Mathematical Finance}, 15(3):425--437, 2005.

\bibitem[KOZ15]{KromerOverbeckZilch2015a}
E.~Kromer, L.~Overbeck, and K.~and Zilch.
\newblock Dynamic systemic risk measures for bounded discrete time processes.
\newblock {\em Preprint}, 2015.

\bibitem[Law10]{DoodFrank}
Public Law.
\newblock {D}odd-{F}rank {W}all {S}treet {R}eform and {C}onsumer {P}rotection
  {A}ct.
\newblock {\em Public Law 111--203}, 2010.

\bibitem[Mur13]{Murphy2013}
D.~Murphy.
\newblock {\em {OTC} {D}erivatives: Bilateral Trading and Central Clearing.}
\newblock Palgrave Macmillan UK, 2013.

\bibitem[Tas00]{Tasche2000}
D.~Tasche.
\newblock Risk {C}ontributions and {P}erformance {M}easurement.
\newblock Technical report, Preprint, Zentrum Mathematik (SCA), 2000.

\bibitem[Tas02]{Tasche2002}
D.~Tasche.
\newblock {\em Expected {S}hortfall and Beyond}, pages 109--123.
\newblock Birkh{\"a}user Basel, 2002.

\bibitem[Tas07]{Tasch2007}
D.~Tasche.
\newblock Euler allocation: Theory and practice.
\newblock {\em Preprint}, 2007.

\bibitem[{US }16]{CFTC2016}
{US CFTC}.
\newblock {S}upervisory {S}tress {T}est of {C}learinghouses.
\newblock {\em A report by staff of the {U.S.} {C}ommodity {F}utures {T}rading
  {C}ommission}, 2016.

\bibitem[VCDF{\etalchar{+}}15]{CORE-2015}
L.A.B.G. Vicente, F.V. Cerezetti, S.R. De~Faria, T.~Iwashita, and O.R. Pereira.
\newblock Managing risk in multi-asset class, multimarket central
  counterparties: The {CORE} approach.
\newblock {\em Journal of {B}anking \& {F}inance}, 2015.

\bibitem[YP17]{YoungPaddrik2017}
H.P. Young and M.~Paddrik.
\newblock {How Safe are Central Counterparties in Derivatives Markets?}
\newblock Technical report, Preprint, 2017.

\end{thebibliography}
\newcommand{\etalchar}[1]{$^{#1}$}

}
\newpage
{\footnotesize
%\section{Glossary}
\printglossaries
}
\end{document}